\documentclass[11pt]{article}

\usepackage{amsmath,amsfonts,amsthm,latexsym,graphicx}
\usepackage{booktabs, color}
\usepackage[latin1]{inputenc}

\usepackage{tikz}
\usetikzlibrary{arrows} 
\usetikzlibrary{patterns}
\usetikzlibrary{positioning}

\newcommand{\norm}[2]{\left\|{#1}\right\|_{#2}}
\newcommand{\abs}[1]{\left|{#1}\right|}
\newcommand{\gap}{{\rm gap}}
\newcommand{\spec}{{\rm spec}}
\renewcommand{\i}{{\mathrm i}}
\DeclareMathAlphabet{\mathbf}{OT1}{cmr}{bx}{it}

\DeclareMathAlphabet{\mathbf}{OT1}{cmr}{bx}{it}

\DeclareMathOperator{\Image}{Im}
\DeclareMathOperator{\dist}{dist}

\newcommand{\vy}{\mathbf y}

\newcommand{\vxi}{\boldsymbol{\xi}}

\newcommand{\bbN}{\mathbb{N}}
\newcommand{\bbR}{\mathbb{R}}
\newcommand{\bbC}{\mathbb{C}}
\newcommand{\bbP}{\mathbb{P}}
\newcommand{\bbE}{\mathbb{E}}

\newcommand{\mc}{\mathcal}

\DeclareMathOperator{\e}{\mathrm e}

\renewcommand{\d}{\mathrm{d}}

\theoremstyle{definition}
\newtheorem{defi}{Definition}

\newtheorem{remark}[defi]{Remark}
\newtheorem{propo}[defi]{Proposition}
\newtheorem{theo}[defi]{Theorem}
\newtheorem{lem}[defi]{Lemma}
\newtheorem{cor}[defi]{Corollary}

\title{On a generalization of the preconditioned Crank-Nicolson Metropolis algorithm}

\author{
Daniel Rudolf
\\
Institut f\"ur Mathematik, Universit\"at Jena\\
Ernst-Abbe-Platz 2, 07743 Jena, Germany\\
email: 
daniel.rudolf@uni-jena.de\\ \\
Bj\"orn Sprungk
\\
Technische Universit\"at Chemnitz\\ 
Reichenhainer Str. 41, 09107 Chemnitz, Germany\\ 
email: 
bjoern.sprungk@mathematik.tu-chemnitz.de
}

\begin{document}
\maketitle
\begin{abstract}
Metropolis algorithms for approximate sampling of probability measures on infinite dimensional Hilbert spaces are  considered and a generalization of the preconditioned Crank-Nicolson (pCN) proposal is introduced. 
The new proposal is able to incorporate information on the measure of interest.   
A numerical simulation of a Bayesian inverse problem indicates that a Metropolis algorithm with such a proposal 
performs independently of the state space dimension and the variance of the observational noise.
Moreover, a qualitative convergence result is provided by a comparison argument for spectral gaps.
In particular, it is shown that the generalization inherits geometric convergence
from the Metropolis algorithm with pCN proposal.
\end{abstract}

\vspace*{2ex}

{\small


}

\section{Introduction}
\label{sec: intro}

Consider a target probability distribution $\mu$ defined 
on a possibly infinite dimensional separable Hilbert space $\mc H$. 
It is of interest to sample from this probability measure
and assumed that there is a density of $\mu$
w.r.t. a Gaussian reference measure $\mu_0$ on $\mc H$
given by
\begin{equation} \label{equ:mu}
	 \frac{\d\mu}{\d\mu_0}(u) = \frac{1}{Z} \exp(-\Phi(u)), \qquad u\in \mc H.
\end{equation}
Here $\Phi \colon \mc H \to \mathbb{R}_+$ 
is a measurable function and $Z= \int_{\mc H} \exp(-\Phi(u))\, \mu_0(\d u)$ 
the normalizing constant.
Such probability measures $\mu$ arise as 
 posterior distributions in Bayesian inference with $\mu_0$ as a Gaussian prior.
Common examples 
in infinite dimensional spaces are inferring spatially 
distributed properties of porous media
or stock prices.

Unfortunately, 
the fact that the normalizing constant $Z$ is typically 
unknown and that $\Phi$ is only available in the form of function evaluations
makes it difficult to sample $\mu$ directly. 
But Markov chains and in particular Metropolis-Hastings (MH) algorithms are applicable
for approximate sampling. 
These algorithms consist of a proposal and an acceptance/rejection step. 
A state is proposed by
a proposal kernel 
but it is only accepted with a certain probability 
which depends on $\frac{\d \mu}{\d \mu_0}$.
The authors of \cite{BeskosEtAl2008} suggested
a modification of a 
Gaussian random walk proposal which is $\mu_0$-reversible. 
The latter property leads to a well-defined MH algorithm in 
infinite dimensional Hilbert spaces, see also \cite{Ti98}.
This proposal was later \cite{CotterEtAl2013} 
referred to as \emph{preconditioned Crank-Nicolson (pCN)} proposal.
Remarkably, the Markov chain of the
resulting pCN Metropolis algorithm
has 
\emph{dimension-independent sampling efficiency}, see \cite{CotterEtAl2013},\cite{HaStVo14}. 
This is a significant advantage compared to earlier,
popular MH algorithms whose performance usually 
deteriorates with increasing state space dimension \cite{CotterEtAl2013},\cite{HaStVo14},\cite{RobertsRosenthal2001}.

We extend the pCN proposal to incorporate information about the target measure $\mu$.
Such an adaption might account for
the anisotropy of the covariance of $\mu$ or the local curvature of $\Phi$.
Intuitively, the resulting Markov chain has on average a larger step size and, thus, explores the state space faster.
This idea is not entirely new.
It is already mentioned in \cite{Ti94} where it is suggested to choose the covariance of the proposal adapted to the target measure.
Later in \cite{GirolamiCalderhead2011} the authors explain how to propose new states using general local metric tensors.
Moreover, in \cite{MartinEtAl2012} the Hessian 
of the negative log density $\Phi$ of $\mu$ is 
employed as local curvature information to design 
a stochastic Newton MH method in finite dimensions
and in \cite{CuiEtAl2014},\cite{Law2014}
a Gauss-Newton variant for capturing global curvature in an 
infinite dimensional setting is outlined.   

Our approach for adapting the proposal to the target measure $\mu$ 
has a similiar motivation as the proposals considered
in \cite{CuiEtAl2014},\cite{Law2014}. It comes from a
local linearization of the unknown-to-observable map in Bayesian inverse problems.
This suggests a particular form for approximating the 
covariance of the 
target measure, 
namely $(C + \Gamma)^{-1}$,
where $C$ denotes the covariance of the
reference measure $\mu_0$ and $\Gamma$ is a 
suitable self-adjoint and positive operator.
We then consider the class of Gaussian proposals with covariance $C_\Gamma = (C + \Gamma)^{-1}$.
By enforcing $\mu_0$-reversibility 
we derive our class of generalized pCN (gpCN) proposal kernels $P_\Gamma$.

In a numerical simulation the resulting Metropolis algorithm seems to perform
independent of dimension and variance. 
Here variance independence refers to the variance of the observational noise, which affects 
the covariance of the target distribution $\mu$.
Particularly, if the variance of the noise decreases the measure 
$\mu$ becomes more concentrated.
Our numerical experiments also indicate that other popular MH 
or random walk algorithms perform worse, i.e., variance dependent.

Moreover, we present a convergence result
for the gpCN Metropolis based on spectral gaps.
It is well known, see \cite{RoRo97}, 
that for Markov chains with reversible transition kernels $K$ a 
strictly positive spectral gap, denoted 
$\gap(K)>0$, is equivalent to a form of geometric ergodicity.
The latter roughly means that, in an appropriate setting, 
the distribution of the $n$th step
of a Markov chain converges exponentially fast to its stationary measure.
We refer to Section~\ref{sec:MC} for precise definitions and further details.

Our main theoretical result, stated in Theorem~\ref{theo:gpCN_conv}, is as follows.
Let us assume that the transition kernel $M_0$ of the pCN algorithm 
has a positive spectral gap, i.e. $\gap(M_0)>0$. 
Then, for any 
$\varepsilon > 0$ 
there is 
an explicitly given probability measure $\mu_R$ such that 
\[
\norm{\mu-\mu_R}{\text{tv}} \leq \varepsilon 
\quad
\text{ and }
\quad
\gap(M_{\Gamma,R})>0
\]
where $M_{\Gamma,R}$ denotes the transition kernel of the gpCN Metropolis algorithm 
targeting the measure $\mu_R$ and 
$\norm{\cdot}{\text{tv}}$ is the total variation 
distance, see \eqref{eq:tv_definition}.

The key for the proof is a new comparison theorem for spectral gaps of Markov chains generated by MH algorithms.
In order to apply this comparison argument we show that the proposal kernels of the pCN and gpCN Metropolis are equivalent and that their Radon-Nikodym derivative 
belongs to an $L_p$-space for a $p>1$.
We note that in \cite{HaStVo14} under additional assumptions on the density function $\frac{\d \mu}{\d \mu_0}$
it is proven that there exists a strictly positive spectral gap of the pCN Metropolis.
Thus, in this setting the gpCN Metropolis algorithm targeting $\mu_R$ also converges exponentially.

The remainder of the paper is organized as follows.
In Section \ref{sec:pre} we state the precise framework, recall preliminary facts, 
and give a brief introduction to Markov chain Monte Carlo 
and MH algorithms including the pCN Metropolis algorithm.
The gpCN Metropolis algorithm is motivated and defined in Section \ref{sec:gpCN}.
Particularly, in Section \ref{subsec: numerics} 
we illustrate its superior 
performance compared to other popular MH algorithms.
In Section \ref{sec:Conv} we state a general result for comparing spectral gaps of MH algorithms and then apply it to the gpCN and pCN Metropolis.
Section~\ref{subsec: loc_gpCN} provides an outlook to gpCN algorithms in 
infinite dimensions which use 
Gaussian proposals with state-dependent covariance.
For the convenience of the reader we 
recall some facts about Gaussian measures 
in Appendix~\ref{sec:Gaussian} and relegate 
more technical proofs to Appendix~\ref{sec:proofs}.

\section{Preliminaries} \label{sec:pre}

\noindent 
Let $\mc H$ be a separable Hilbert space with
inner-product and norm denoted
by $\langle \cdot, \cdot \rangle$ and $\|\cdot\|$. 
By $\mathcal{B}(\mc H)$ we denote the corresponding Borel $\sigma$-algebra 
and by $\mc L(\mc H)$ the set of all bounded, linear operators 
$A\colon \mc H \to \mc H$.
Further, we have a Gaussian measure $\mu_0 = N(0, C)$
on $(\mc H,\mathcal{B}(\mc H))$. 
Here and in the remainder of the 
paper $C\colon \mc H \to \mc H$ denotes a nonsingular \emph{covariance operator}
on $\mc H$, i.e., a bounded, self-adjoint and positive trace class operator with $\ker C= \{0\}$.
By $\mu$ we denote the probability measure 
of interest on $(\mc H,\mathcal{B}(\mc H))$ given through 
the density defined in \eqref{equ:mu}.
Typically, the desired distribution is complicated and the density only 
known up to a constant, 
which makes direct sampling from 
$\mu$ difficult.
This is the reason why Markov chains are 
used for approximate sampling according to $\mu$.

\subsection{Markov chains and spectral gaps}\label{sec:MC}
We give a short introduction to Markov chains and Markov chain Monte Carlo (MCMC)
methods on general state spaces.
We call a mapping $K \colon \mc H \times \mathcal{B}(\mc H) \to [0,1]$ a \emph{transition kernel}, 
if $K(x,\cdot)$ is a probability measure on $(\mc H,\mathcal{B}(\mc H))$ 
for each $x\in \mc H$ and $K(\cdot,A)$ is a measurable function for each $A\in\mathcal{B}(\mc H)$.
Then, a \emph{Markov chain with transition kernel $K$} 
is a sequence of random variables $(X_n)_{n\in \bbN}$, mapping from some probability
space $(\Omega,\mathcal{F},\mathbb{P})$ to $(\mc H,\mathcal{B}(\mc H))$,
satisfying
\[
 \mathbb{P}(X_{n+1} \in A \mid X_1,\dots,X_n) 
 = \mathbb{P}(X_{n+1} \in A \mid X_n)
 = K(X_n,A)
\]
almost surely for all $A\in \mathcal{B}(\mc H)$.
Most properties of a Markov chain can be expressed as properties of its 
transition kernel.
For example, 
we say the transition kernel $K$ is \emph{$\mu$-reversible} if 
\begin{equation}  \label{eq: det_balance_gen}
 	K(u, \d v) \, \mu(\d u) = K(v, \d u) \, \mu(\d v)
\end{equation}
in the sense of measures on $\mc H\times \mc H$.
This property is also known as the
\emph{detailed balance condition} and it implies 
that the distribution $\mu$ is a stationary 
or invariant probability measure of a 
Markov chain with transition kernel $K$, i.e., if $X_1 \sim \mu$ then also $X_{2} \sim \mu$.

Each $\mu$-reversible transition kernel $K$ on $(\mc H,\mathcal{B}(\mc H))$ 
induces a \emph{Markov operator}, which we shall also denote by $K$, given by
\[
 Kf(u) = \int_{\mc H} f(v)\, K(u,\d v), \quad f\in L_2(\mu),
\]
where
\[
 L_2(\mu) = \left\{ f\colon \mc H \to \bbR
 \mid \| f \|_{2,\mu} 
 := \left( \int_{\mc H} |f(u)|^2\,\mu(\d u) \right)^{1/2} <\infty \right\},
\]
is the Hilbert space of measurable, square integrable functions with respect to
$\mu$.
By the $\mu$-reversibility we have that $K\colon L_2(\mu)\to L_2(\mu)$ is a 
bounded and self-adjoint linear operator.
We also introduce the closed subspace
\[
 L_2^0(\mu) = \left\{ f\in L_2(\mu) \mid \int_{\mc H} f(u)\, \mu(\d u)=0 \right\}
\]
of $L_2(\mu)$
and the operator norm 
\[
 \norm{K}{\mu} = 
 \sup_{f\in L_2^0(\mu),\, f\not =0} 
 \frac{\norm{Kf}{2,\mu}}{\norm{f}{2,\mu}}
\]
for $K: L_2^0(\mu) \to  L_2^0(\mu)$.
Let $\spec(K\,|\, L_2^0(\mu))$ denote 
the spectrum of $K$ on $L_2^0(\mu)$. Then, we also have
\[
 \norm{K}{\mu} = \sup\{\abs{\lambda}\colon \lambda \in \spec(K\,|\, L_2^0(\mu))\}.
\]
We define the \emph{spectral gap of $K$ (w.r.t. $\mu$) by} $\gap(K) = 1- \norm{K}{\mu}$. 
This is an important quantity which can be used to formulate 
conditions ensuring an exponentially fast convergence of the distribution of $X_n$ to $\mu$. 
To be more precise, we introduce the total variation distance of two
probability measures $\nu_1,\nu_2$ on $(\mc H,\mathcal{B}(\mc H))$
by
\begin{equation}  \label{eq:tv_definition}
  \norm{\nu_1 - \nu_2}{\text{tv}} 
 := \sup_{A\in\mathcal{B}(\mc H)} \abs{\nu_1(A)-\nu_2(A)}.
\end{equation}
Let $\nu$ be the initial distribution of our Markov chain, i.e., $X_1\sim \nu$.
Then,
with 
\[
	K^{n}(u,A) = \int_{\mc H} K^{n-1}(v, A)\,K(u, \d v),\quad A\in \mc B(\mc H),
\]
for $n\in \mathbb{N}$, the distribution of $X_{n+1}$ is given by
\[
 \nu K^n(A) = \int_{\mc H} K^n(u,A)\, \nu(\d u).
\]

In the setting above it is well known, see \cite[Proposition~2.2]{RoRo97}, 
that $\norm{K}{\mu}<1$, or equivalently $\gap(K)>0$, 
holds, iff the transition kernel is $L_2(\mu)$-geometrically ergodic. 
Here by $L_2(\mu)$-geometric ergodicity we mean that,
there exists a number $r\in [0,1)$ such that
for any probability measure $\nu$, which has a density 
$\frac{\d\nu}{\d\mu}\in L_2(\mu)$ w.r.t $\mu$,
there is a constant $C_\nu<\infty$ such that
\[
 \norm{\nu K^n - \mu}{\text{tv}} \leq C_\nu\, r^n, \qquad n\in\mathbb{N}.
\]
If the distribution of $X_{n}$ converges to $\mu$, then the Markov chain $(X_n)_{n\in\bbN}$ 
can be used for approximate sampling from 
$\mu$.
This leads to Markov chain Monte Carlo methods for the computation of expectations. 
The mean $\mathbb{E}_\mu(f)$ of a function $f\colon \mc H \to \bbR$ w.r.t $\mu$ 
can then be approximated by the time average
\[
	S_{n,n_0}(f) = \frac{1}{n} \sum_{j=1}^n f(X_{j+n_0})
\]
where $n$ is the sample size and $n_0$ a burn-in parameter to decrease the influence of the initial distribution.
The spectral gap of $K$ of the Markov chain $(X_n)_{n\in \bbN}$ can then be applied to assess the error of the time average $S_{n,n_0}(f)$.
We assume $\gap(K)>0$ and mention two results.
The first is rather classical and due to Kipnis and Varadhan \cite{Kipnis86}. 
If the initial distribution is $\mu$ and $f\in L_2(\mu)$, then the error 
$\sqrt{n}(S_{n,n_0}(f)-\mathbb{E}_{\mu}(f))$ converges weakly to $N(0,\sigma_{f,K}^2)$ with
\[
	\sigma_{f,K}^2 
	= \langle (I+K)(I-K)^{-1}(f-\mathbb{E}_\mu(f)),(f-\mathbb{E}_\mu(f)) \rangle_{\mu}
	\leq \frac{2 \norm{f}{2,\mu}^2}{\gap(K)}
\]
where $\langle \cdot,\cdot \rangle_{\mu}$ denotes the inner-product in $L_2(\mu)$.
The second result is more recent and provides a non-asymptotic bound for 
the mean square error. We have
\[
	\sup_{\norm{f}{4}\leq 1} \mathbb{E}\abs{S_{n,n_0}(f)-\mathbb{E}_\mu(f)}^2
	\leq \frac{2}{n \cdot \gap(K)} + \frac{C_\nu \norm{K}{\mu}^{n_0}}{n^2 \cdot \gap(K)^2}
\]
with $\norm{f}{4} = \left( \int_{\mc H} \abs{f(u)}^4\, \mu(\d u)\right)^{1/4}$ and a number $C_\nu\geq 0$ depending on the initial distribution $\nu$. 
We refer to \cite{Ru12} for 
details.\par
This shows that $\gap(K)$ is a crucial quantity in the study of Markov chains and the numerical analysis of MCMC methods.

\subsection{Metropolis algorithm with pCN proposal} \label{sec: pCN}
In this work we focus on Markov chains generated by 
the Metropolis algorithm.
This algorithm employs a transition kernel on $(\mc H,\mc B(\mc H))$ for proposing new states which we shall denote by $P$ and call \emph{proposal kernel}.
Moreover, let $\alpha\colon \mc H \times \mc H \to [0,1]$ be a measurable function denoting the \emph{acceptance probability}.
Then, a transition of a Markov chain $(X_n)_{n\in \bbN}$ generated by the Metropolis algorithm can be represented in algorithmic form:
\begin{enumerate}
 \item 
 Given the current state $X_n = u$, draw independently a sample $v$ of a random variable $V\sim P(u,\cdot)$ and a sample $a$ of a random variable $A\sim \text{Unif}[0,1]$.
 \item 
 If $a < \alpha(u,v)$, then set $X_{n+1} = v$, otherwise set $X_{n+1} = u$.
\end{enumerate}
The transition kernel of such a Markov chain is then
\begin{equation} \label{eq: metro_kern}
 	M(u, \d v) =  \alpha(u,v) P(u,\d v) 
	+\delta_u(\d v) \,\int_{\mc H} (1 - \alpha(u, w)) \, P(u,\d w)
\end{equation}
and we call it \emph{Metropolis kernel}. 
It is well known, see \cite{Ti98}, that $M$ is reversible w.r.t. 
$\mu$ if $\alpha(\cdot,\cdot)$ is chosen as
\begin{align} \label{al: acc_prob}
	\alpha(u,v) = \min\left\{ 1 , \frac{\d \eta^\bot}{\d \eta}(u,v)\right\}, \qquad u,v\in \mc H,
\end{align}
where $\frac{\d \eta^\bot}{\d \eta}$ denotes the Radon-Nikodym derivative of the measures
\begin{align*}
   \eta(\d u,\d v) & := P(u, \d v) \, \mu(\d u) \qquad \mbox{and} \qquad \eta^\bot (\d u,\d v) := P(v, \d u) \, \mu(\d v),
\end{align*}
which we assume to exist.
For finite dimensional state spaces
the condition of absolute continuity  
of $\eta^\bot$ w.r.t. 
$\eta$ is often easily satisfied.
However, for infinite dimensional state spaces
this becomes a real issue, since there measures tend to be mutually singular.
As pointed out in \cite{BeskosEtAl2008},\cite{CotterEtAl2013} a possible way to ensure the existence of $\frac{\d \eta^\bot}{\d \eta}$ is to choose a proposal kernel $P$ which is $\mu_0$-reversible, i.e.,
\begin{equation}  \label{eq: det_balance}
 	P(u, \d v) \, \mu_0(\d u) = P(v, \d u) \, \mu_0(\d v).
\end{equation}
Then, due to the fact that $\frac{\d \mu}{\d \mu_0}$ and $\frac{\d \mu_0}{\d \mu}$ exist,
see \eqref{equ:mu}, it follows that
\begin{align}
	\frac{\d \eta^\bot}{\d \eta}(u,v) = 
	\frac{\d \mu}{\d\mu_0}(v) \frac{\d \mu_0}{\d\mu}(u)
	= \exp( \Phi(u) - \Phi(v))
\end{align}
and, hence, $\alpha(u,v) = \min\left\{1, \exp( \Phi(u) - \Phi(v))\right\}$.

We next introduce  
the Metropolis algorithm with the \emph{preconditioned Crank-Nicolson} (pCN) proposal, 
see also \cite{CotterEtAl2013} for details.
The pCN proposal kernel arises from a discretization of an Ornstein-Uhlenbeck process with invariant measure $\mu_0$
and takes the form
\begin{align}\label{al: pCN}
	P_{0}(u,\cdot) = N(\sqrt{1-s^2}u,s^2 C)
\end{align} 
where $s\in[0,1]$ denotes a variance or stepsize parameter.
It is straightforward to verify that $P_{0}$ is $\mu_0$-reversible.
Namely, by applying \eqref{equ:Gaussian_affine} from 
Appendix \ref{sec:Gaussian} we deduce
\[
	P_{0}(u, \d v) \, \mu_0(\d u) 
	= N\left( \begin{bmatrix}  0\\ 0 \end{bmatrix}, \begin{bmatrix}  C & \sqrt{1-s^2}C \\ \sqrt{1-s^2}C & C \end{bmatrix} \right)
	= P_{0}(v, \d u) \, \mu_0(\d v).
\]
In the following we call the resulting Metropolis algorithm with proposal $P_{0}$
simply pCN Metropolis algorithm or 
pCN Metropolis and denote its Metropolis kernel by $M_{0}$.

Next, we generalize the pCN Metropolis algorithm to allow for proposal kernels which employ a different covariance structure than the covariance of $\mu_0$.

\section{Metropolis with gpCN proposals}  \label{sec:gpCN}
In recent years many authors have proposed and pursued the idea to construct proposals 
which try to exploit certain geometrical features of the target measure, see for example \cite{GirolamiCalderhead2011},\cite{MartinEtAl2012},\cite{Law2014},\cite{CuiEtAl2014}.
\par
We consider generalized pCN (gpCN) proposals 
which aim to adapt to the covariance structure of the target measure $\mu$.
We motivate our gpCN proposal, show that it is well-defined in function spaces 
and illustrate its superior  
performance in a simple but common setting.

\subsection{Motivation from Bayesian inference} \label{subsec: Motiv}
We briefly recall the Bayesian framework for inverse problems and refer to \cite{ErnstEtAl2015} for an overview 
and to \cite{Stuart2010} for a comprehensive introduction to the topic.

Assume $X$ is  a random variable on $(\mc H,\mathcal{B}(\mc H))$ 
with distribution $\mu_0 = N(0,C)$.
Here $\mu_0$ is called the \emph{prior} distribution and describes our initial uncertainty 
about $X$.
Let $Y$ be a random variable on $\mathbb{R}^m$ given by
\begin{equation} \label{eq: bay_model}
  Y = G(X) + \varepsilon
\end{equation}
with a continuous map $G\colon \mc H\to \mathbb{R}^m$ 
and $\varepsilon \sim N(0,\Sigma)$, independent of $X$, with $\Sigma \in \bbR^{m\times m}$.
The variable $Y$ models an observable quantity depending on  
$X$ via the map $G$ which is perturbed by additive noise $\varepsilon$.
Then, given some observation $y\in \mathbb{R}^m$ of $Y$ 
we want to infer $X$, i.e., we are interested in the conditional distribution of $X$ given the event $Y=y$.
We denote this conditional distribution by $\mu$ and call it \emph{posterior} distribution.
In particular, in this setting $\mu$
admits a representation of the form \eqref{equ:mu} with
\begin{equation}\label{equ:Phi}
	\Phi(u) = \frac 12 |y- G(u)|^2_{\Sigma^{-1}}
\end{equation}
where $\abs{x}_{\Sigma^{-1}}^2 = x^T \Sigma^{-1} x $ for $x\in \bbR^m$.

A special situation appears if $G(u)=L u + b$ with a linear mapping 
$L\colon \mc H \to \mathbb{R}^m$ and $b\in \mathbb{R}^m$. 
Then, it is known from \cite{Mandelbaum1984} that $\mu= N(m,\widehat C)$ with
\begin{equation}  \label{eq: target_normal}
  m = C L^*(LCL^*+\Sigma)^{-1} (y-b), 
  \qquad \widehat C=(C^{-1}+L^*\Sigma^{-1}L)^{-1},
\end{equation}
where $L^*$ denotes the adjoint operator of $L$.
If we want to sample approximately from a Gaussian 
target measure $\mu = N(m,\widehat C)$ by Metropolis 
algorithms with Gaussian proposals, it seems beneficial to 
employ $s^2 \widehat C$ as proposal covariance, 
see for example \cite{Ti94},\cite{RobertsRosenthal2001},\cite{Law2014}.
Intuitively, since then the Gaussian proposal possesses the same principal directions and the same ratio of variances as the Gaussian target measure, the proposed states should be accepted more often than for other proposals. See also Figure~\ref{Fig: gauss_prop} for an illustration. 
This leads to a higher average acceptance probability 
and, thus, a faster exploration of the state space.

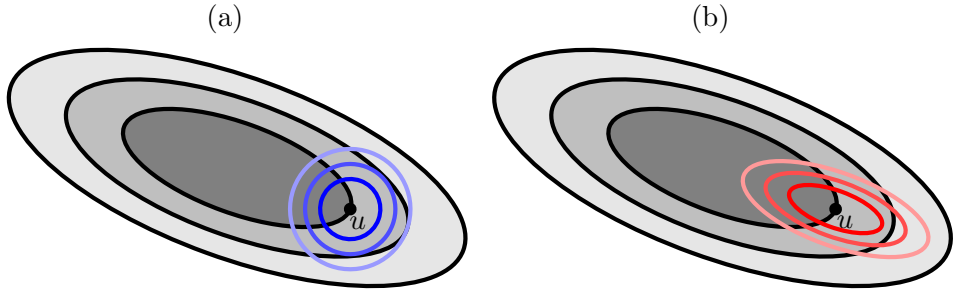
\begin{figure}[htb]
\centering
\def\target{(0,0) ellipse (3 and 8)} 
\def\targetZ{(0,0) ellipse (4.5 and 8*4.5/3)} 
\def\targetZZ{(0,0) ellipse (6 and 8*6/3)} 
\def\propball{(0,8) ellipse (2 and 2)} 
\def\propballZ{(0,8) ellipse (3 and 3)} 
\def\propballZZ{(0,8) ellipse (4 and 4)} 
\def\proptarg{(0,8) ellipse (3/2.45 and 8/2.45)} 
\def\proptargZ{(0,8) ellipse (3/1.63 and 8/1.63)} 
\def\proptargZZ{(0,8) ellipse (3/1.23 and 8/1.23)} 

\tikzstyle{P_1} = [draw,black,ultra thick]
\tikzstyle{P_21} = [draw,blue!100,ultra thick]
\tikzstyle{P_22} = [draw,blue!70,ultra thick]
\tikzstyle{P_23} = [draw,blue!40,ultra thick]
\tikzstyle{P_31} = [draw,red!100,ultra thick]
\tikzstyle{P_32} = [draw,red!70,ultra thick]
\tikzstyle{P_33} = [draw,red!40,ultra thick]

\tikzset{
  pat1/.style={pattern=horizontal lines,pattern color=#1},
  pat1/.default=red
}

\tikzset{
  pat2/.style={pattern=vertical lines,pattern color=#1},
  pat2/.default=black
}
 \begin{minipage}{0.49\textwidth} 
 \centering
 (a)\\
\begin{tikzpicture}[scale=0.2,>=latex']
    \begin{scope}[rotate=-110]
	\fill[fill=black!10] \targetZZ;
	\fill[fill=black!25] \targetZ;
	\fill[fill=black!50] \target;

	\fill(0,8)circle(12pt);
	\node at (0.7,8.8) {$u$};
      
	\path[P_1] \target;
	\path[P_1] \targetZ;
	\path[P_1] \targetZZ;
        \path[P_21] \propball;
        \path[P_22] \propballZ;
        \path[P_23] \propballZZ;
    \end{scope}
 \end{tikzpicture}
\end{minipage}
\hfill
 \begin{minipage}{0.49\textwidth} 
 \centering
 (b)\\
\begin{tikzpicture}[scale=0.2,>=latex']
    \begin{scope}[rotate=-110]
	\fill[fill=black!10] \targetZZ;
	\fill[fill=black!25] \targetZ;
	\fill[fill=black!50] \target;

	\fill(0,8)circle(12pt);
	\node at (0.6,8.9) {$u$};
	\path[P_1] \target;
	\path[P_1] \targetZ;
	\path[P_1] \targetZZ;
        \path[P_31] \proptarg;
        \path[P_32] \proptargZ;
        \path[P_33] \proptargZZ;

    \end{scope}
 \end{tikzpicture}
\end{minipage}
\caption{\label{Fig: gauss_prop}
For a Gaussian target measure 
$\mu=N(m,\widehat{C})$ 
and current state $u$ the 
region of acceptance $\{v: \alpha(u,v) = 1\}$ 
(dark grey region) as well as two regions of possible 
rejection $\{v: \underline{p} \leq \alpha(u,v) < \overline{p} \leq 1\}$ (lighter grey regions) are displayed.
Moreover, we present the contour lines (blue and red, resp.) 
of Gaussian proposals $N(u,s^2C)$ with covariance $C=I$ 
in part (a) and target covariance $C=\widehat C$ in part (b).
}
\label{fig:motiv_gpCN}
\end{figure}

The affine case indicates how we can construct good Gaussian proposal kernels if the map $G$ is nonlinear but smooth.
For a fixed $u_0 \in \mc H$ local linerization leads to
\[
	G(u) = G(u_0) + \nabla G(u_0)\,(u-u_0) + r(u)
\]
with a remainder 
term $r(u) \in \mathbb{R}^m$. 
For a sufficiently smooth $G$ the remainder 
$r$ is small (in a neighborhood of $u_0$), so that
\[
 \widetilde G(u) = G(u_0) + \nabla G(u_0)\,(u-u_0)
\]
is close to $G(u)$ (in a neighborhood of $u_0$).
The substitution of $G$ by $\widetilde G$ 
in the model \eqref{eq: bay_model} leads to a 
Gaussian target measure $\widetilde \mu = N(\widetilde m, \widetilde C)$ 
with covariance 
\[
 \widetilde C = (C^{-1} + L^* \Sigma^{-1} L)^{-1}, 
 \qquad
 L = \nabla G(u_0).
\]
By the fact that $G$ and $\widetilde G$ are close,  
we also have that the measures $\mu$ and $\widetilde \mu$ are close as well. 
Then, it is reasonable to use $\widetilde C$ in the covariance operator of the proposal in a Metropolis algorithm.
Of course, there might be other choices besides 
a simple linearization of $G$ at one point.
For example, averaging linearizations at several points $u_1,\dots,u_n \in \mc H$ leads to
\[
	\widetilde C = \Big(C^{-1} + \frac 1N \sum_{n=1}^N L_n^* \, 
	\Sigma^{-1}L_n\Big)^{-1}, \qquad L_n = \nabla G(u_n).
\]
Natural candidates for the points $u_1,\dots,u_N$ are 
samples according to the prior or 
samples taken from a short run of a preliminary 
Markov chain with the posterior as stationary measure,
cf. the adaptive method in \cite[Section 3.4]{CuiEtAl2014}. 
One could also think of a state-dependent covariance $C(u)$. 
This motivates the study of 
proposals which use covariances of the form
$
 C_{\Gamma} = (C^{-1} + \Gamma)^{-1}
$
for suitably chosen operators $\Gamma$.

\subsection{Well-defined gpCN proposals} \label{sec:well_gpCN}
In this section we introduce the gpCN proposal kernel 
and prove that the Metropolis algorithm with this proposal is well-defined 
in the sense that it leads to a $\mu$-reversible transition kernel.

For this we introduce the set $\mc L_+(\mc H)$ of all bounded, 
self-adjoint and positive linear operators $\Gamma: \mc H \to \mc H$.
We define the operators 
\begin{equation} \label{equ:tildeC}
	C_\Gamma := (C^{-1} + \Gamma )^{-1}, \qquad \Gamma \in \mc L_+(\mc H),
\end{equation}
motivated in Section~\ref{subsec: Motiv}, where $C$ denotes the covariance operator of the prior measure $\mu_0 = N(0,C)$, for which we also use the equivalent representation 
\begin{equation} \label{equ:tildeC_2}
	C_\Gamma = C^{1/2} \; ( I + H_\Gamma)^{-1} \; C^{1/2}, 
 	\qquad H_\Gamma:=C^{1/2} \Gamma C^{1/2}.
\end{equation}
In the following we prove that $C_\Gamma$ can
be considered as covariance operator.

\begin{propo} \label{propo:tildeC}
Let $C$ be a nonsingular covariance operator 
on $\mc H$, $\Gamma\in \mc L_+(\mc H)$ 
and $C_\Gamma$ with $H_\Gamma$ given as in \eqref{equ:tildeC_2}.
Then $H_\Gamma \in \mc L_+(\mc H)$ is trace class and $C_\Gamma$ is also a nonsingular covariance operator on $\mc H$.
\end{propo}

\begin{proof}
That $H_\Gamma \in \mc L_+(\mc H)$ follows by construction.
Furthermore, since $H_\Gamma$ is a composition of two Hilbert-Schmidt 
and one bounded operator, $C^{1/2}$ and $\Gamma$, respectively, 
it is trace class \cite[Proposition~1.1.2]{DaPrato2004}.
Since $H_\Gamma$ is 
selfadjoint and compact, we have from Fredholm operator theory that the 
operator $I+H_\Gamma$ is invertible iff ${\rm ker}\,H_\Gamma = \{0\}$.
The latter is the case since $H_\Gamma$ is positive which implies $\langle (I+H_\Gamma)u, u\rangle \geq \langle u,u \rangle$.
Hence, the inverse $(I+H_\Gamma)^{-1}$ exists and, moreover, $(I+H_\Gamma)^{-1} \in \mc L_+(\mc H)$ with $\|(I+H_\Gamma)^{-1}\|\leq 1$.
The self-adjointness and positivity of $C_\Gamma$ follows immediately and since $C_\Gamma$ is a composition of two nonsingular Hilbert-Schmidt operators and a nonsingular bounded operator, $C^{1/2}$ and $( I + H)^{-1}$, respectively, it is trace class and nonsingular as well.
\end{proof}

By Proposition \ref{propo:tildeC} we can use the covariance operator $C_\Gamma$ for constructing proposal kernels. Specifically, we consider
\begin{equation} \label{equ:gpCN_ansatz}
	P(u, \cdot) = N( Au, s^2 C_\Gamma), \qquad s \in [0,1), \; \Gamma \in \mc L_+(\mc H),
\end{equation}
where $A\colon \mc H \to \mc H$ denotes a suitably chosen bounded linear operator on $\mc H$. 
Here $A$ should be chosen such that $P$ is $\mu_0$-reversible, 
which means that a Metropolis kernel with proposal $P$ is $\mu$-reversible, 
see Section~\ref{sec: pCN}.
By applying \eqref{equ:Gaussian_affine} we obtain in this setting 
\[
	P(u, \d v)\, \mu_0(\d u) 
	= N\left( \begin{bmatrix}  0\\ 0 \end{bmatrix}, \begin{bmatrix}  C 
	& CA^* \\ AC & ACA^* + s^2C_\Gamma\end{bmatrix} \right)
\]
and
\[
       P(v, \d u)\, \mu_0(\d v)  
	= N\left( \begin{bmatrix}  0\\ 0 \end{bmatrix}, 
	\begin{bmatrix}  ACA^* + s^2 C_\Gamma & AC \\ CA^* & C\end{bmatrix} \right).
\]
Thus, for satisfying \eqref{eq: det_balance} we need to choose $A$ 
so that
\begin{equation} \label{equ:A_cond}
	AC = CA^*, \qquad ACA^* +s^2 C_\Gamma = C.
\end{equation}
By straightforward calculation we obtain as the formal solution to \eqref{equ:A_cond}
\begin{equation} \label{equ:A}
	A = A_{\Gamma} = C^{1/2} \; \sqrt{I - s^2 \left( I + H_\Gamma \right)^{-1} }  C^{-1/2}.
\end{equation}
The following lemma ensures that this choice of 
$A$ yields a well-defined bounded linear operator on $\mc H$.

\begin{lem} \label{lem:A}
Let the assumptions of Proposition \ref{propo:tildeC} be satisfied and let $s\in[0,1)$.
Then \eqref{equ:A} defines a bounded linear operator $A_{\Gamma}:\mathrm{Im}\, C^{1/2} \to \mc H$.
\end{lem}
The well-definedness of $A_{\Gamma}:\mathrm{Im}\, C^{1/2} \to \mc H$ follows rather easily whereas its boundedness is not trivial. Namely, one easily can construct a bounded $B \in \mc L(\mc H)$ such that $C^{1/2}B C^{-1/2}$ is unbounded on $\mathrm{Im}\, C^{1/2}$.
Since the proof of Lemma \ref{lem:A} is rather technical, it is postponed to Appendix \ref{sec:proof_lem:A}.

Lemma \ref{lem:A} allows us now to extend $A_\Gamma$ to $\mc H$ by continuation, because the Cameron-Martin space $\mathrm{Im}\, C^{1/2}$ is a dense subspace of $\mc H$. 
For simplicity we denote this continuous extension again by $A_\Gamma:\mc H\to\mc H$.

\begin{defi}[gpCN proposal] \label{def:gpCN_proposal}
 For $s \in [0,1)$ and $\Gamma \in \mc L_+(\mc H)$ 
 the \emph{generalized pCN proposal kernel} is given by
\begin{equation} \label{equ:gpCN}
	P_{\Gamma}(u, \cdot) := N(A_{\Gamma} u, s^2 C_\Gamma).
\end{equation}
\end{defi}
For the zero operator $\Gamma = 0$ we recover  
the pCN proposal.
By Lemma~\ref{lem:A} and the arguments given 
in Section~\ref{sec: pCN}
we obtain the following important result.
\begin{cor} \label{cor:gpCN}
Let $\mu_0 = N(0,C)$ 
and $\mu$ be given by \eqref{equ:mu}.
Let the assumptions of Lemma \ref{lem:A} be satisfied.
Then, a gpCN proposal kernel $P_{\Gamma}$ given by \eqref{equ:gpCN} 
and an acceptance probability 
$\alpha(u,v) = \min\left\{1, \exp( \Phi(u) - \Phi(v))\right\}$ 
induce a $\mu$-reversible Metropolis kernel denoted by $M_{\Gamma}$.
\end{cor}
For simplicity we also call the Metropolis algorithm with transition kernel $M_{\Gamma}$ just gpCN Metropolis.
There are connections of the gpCN Metropolis to other 
recently developed Metropolis algorithms for general Hilbert spaces which also use more sophisticated choices for the proposal than the pCN proposal.
The following two remarks address these connections.

\begin{remark}
The gpCN proposals form a subclass of the 
\emph{operator weighted proposals} introduced in \cite{CuiEtAl2014},\cite{Law2014}.
The particular form of the gpCN proposal 
allows us to derive properties such as boundedness of the ``proposal mean operator'' 
$A_\Gamma$ and the convergence of 
the resulting Markov chain, see Section \ref{sec:Conv}.
These issues were left open in \cite{CuiEtAl2014},\cite{Law2014}.
\end{remark}

\begin{remark}
In \cite{PinskiEtAl2014} the authors compute a Gaussian measure $\mu_* = N(m_*, C_*)$ which comes closest 
to $\mu$ w.r.t. the Kullback-Leibler distance.
The admissible class of Gaussian measures 
considered there is closely related to our 
parametrized proposal covariances $C_\Gamma$, 
although their class of Gaussian measures is slightly larger.
The measure $\mu_*$ is then used to construct a proposal kernel 
$P_*(u,\cdot) = N(m_* + \sqrt{1-s^2}(u-m_*),s^2C_*)$ for Metropolis algorithms. 
Note that $P_*$ is not $\mu_0$-reversible but $\mu_*$-reversible, 
since it is a pCN proposal given 
the prior $\mu_*$.
In order to obtain a $\mu$-reversible Metropolis 
kernel the authors need to adapt the acceptance probability 
by including terms of $\frac{\d \mu_*}{\d \mu_0}$, cf. Section \ref{subsec: loc_gpCN}.
Thus, the authors of \cite{PinskiEtAl2014} also use a different covariance operator than the prior covariance in a pCN proposal in order to increase the efficiency of the Metropolis algorithm.
The difference to our approach is the way they ensure the $\mu$-reversibility of the algorithm.
They keep the mean of the original pCN proposal and modify the acceptance probability whereas we modify also the mean of the proposal to maintain its $\mu_0$-reversibility and, therefore, can leave the acceptance probability unchanged.
\end{remark}

\subsection{Numerical illustrations} \label{subsec: numerics}

We illustrate the gpCN Metropolis 
algorithm for approximating samples of a posterior distribution in Bayesian inference.
In particular, we compare different Metropolis algorithms 
and investigate which of these 
perform independently of the state space dimension and of the variance of the involved noise. 

We consider the same setting and inference problem as in \cite[Section 6.1]{PinskiEtAl2014}. 
Assume noisy observations $y_j = p(0.2j) + \varepsilon_j$ with $j=1,\ldots,4$, of the solution $p$ of
\begin{equation} \label{equ:PDE}
	\frac \d{\d x}\left( \e^{u(x)}\, \frac \d{\d x}p(x)\right) = 0, \qquad p(0) = 0,\; p(1) = 2,
\end{equation}
on $D=[0,1]$ are given and we want to infer $u$.
Here the $\varepsilon_j$ are independent realizations of
the normal distribution $N(0, \sigma^2_\varepsilon)$. 
We place a Gaussian prior $N(0,\Delta^{-1})$ 
with $\Delta = \frac {\d^2}{\d x^2}$ on the completion $\mc H_c$ of $H_0^1(D) \cap H^2(D)$ in $L^2(D)$.
Recall that $(\Omega,\mathcal{F},\mathbb{P})$ is a probability space 
and let $U\colon\Omega \to  \mc H_c \subset L^2(D)$
be a random function with  distribution $N(0,\Delta^{-1})$.
This allows us to represent the random function $U$ as
\begin{equation} \label{equ:KLE}
U(\omega)(x)
	= \frac{\sqrt 2}{\pi} \sum_{k=1}^\infty \xi_k(\omega) \sin(k \pi x), 
	\qquad \xi_k \sim N(0,k^{-2}),
\end{equation}
$\bbP$-a.s. where all random variables $\xi_k$ are independent.
Thus, inference on $u$ is equivalent to inference on $\vxi = (\xi_k)_{k\in\bbN}$.
This leads to the prior $\mu_0$ for $\vxi$ on $\mc H := \ell^2(\bbR)$ given by $\mu_0 = N(0,C)$ with
$C=\mathrm{diag}\{k^{-2}: k\in\bbN\}$.
Further, we denote by $\mu$ the resulting conditional distribution of $\vxi$ 
given the observed data $y_1,\dots,y_4$.
The measure $\mu$ is given by a density of the form \eqref{equ:mu} with $\Phi$ as in \eqref{equ:Phi} where $\Sigma = \sigma^2_\varepsilon I$ and $G(\vxi)$ is the mapping
\[
 \vxi \mapsto u(\cdot, \vxi) \mapsto p(\cdot, \vxi) \mapsto (p(0.2j, \vxi))_{j=1}^4.
\]

We test the performance of $\mu$-reversible Metropolis algorithms for computing expectations w.r.t. $\mu$ of a function
$f\colon\ell^2(\bbR)\to\bbR$.
We consider four Metropolis algorithms denoted by RW, pCN, GN-RW and 
gpCN with different proposal kernels:
\begin{itemize}
\item
RW: Gaussian random walk proposal $P_1(\vxi, \cdot) = N(\vxi, s^2 C)$,

\item
pCN: pCN proposal $P_2(\vxi, \cdot) = N(\sqrt{1-s^2}\vxi, s^2 C)$,

\item
GN-RW: Gauss-Newton random walk proposal $P_3(\vxi, \cdot) = N(\vxi, s^2 C_\Gamma)$,

\item
gpCN: gpCN proposal $P_4(\vxi, \cdot) = N(A_{\Gamma}\vxi, s^2 C_\Gamma)$.
\end{itemize}
Here we choose $\Gamma = \sigma_{\varepsilon}^{-2} LL^\top$ 
with $L = \nabla G(\vxi_\mathrm{MAP})$ and
\[
	\vxi_\mathrm{MAP} = 
	\operatornamewithlimits{argmin}_{\xi\in \Image C^{1/2}} 
	\left(\sigma_\varepsilon^{-2} 
	\abs{\vy - G(\vxi)}^2 + \| C^{-1/2} \vxi\|^2 \right).
\]
The solution of \eqref{equ:PDE} is given by 
$p(x) = 2 S_x(\e^{-u})/S_1(\e^{-u})$ with $S_x(f) = \int_0^x f(y) \d y$ and, thus, the gradient $\nabla G(\vxi)$ can be easily computed by differentiating the explicit formula for $p$ w.r.t. $\vxi$.\footnote{In general elliptic PDEs can be solved in a weak sense by variational methods. 
Then adjoint methods known from PDE constrained optimization and parameter identification can be employed to compute $\nabla G(\vxi)$, 
see \cite[Chapter 6]{Vogel2002} for details.}
Furthermore, we apply the Levenberg-Marquardt algorithm to solve the above optimization problem for the MAP estimator $\vxi_\mathrm{MAP}$.
For all Metropolis algorithms we tune $s$ such that the average acceptance rate is about $0.25$\footnote{
The empirical performance of each algorithm was best for this particular tuning.}.
As a metric for comparison we consider and estimate 
the \emph{effective sample size} 
\[
 \text{ESS} = \text{ESS}(n,f,(\vxi_k)_{k\in\bbN}) = n\left[1+2\sum_{k\geq 0}\gamma_f(k)\right]^{-1}.
\]
Here $n$ is the number of samples taken from a Markov chain $(\vxi_k)_{k\in\bbN}$ 
with, say, a Metropolis transition kernel $M$ 
and $\gamma_f$ denotes the autocorrelation 
function $\gamma_f(k) = \mathrm{Corr}(f(\vxi_{n_0}), f(\vxi_{n_0+k}))$ 
for a quantity of interest $f$.

The value of $\text{ESS}$ corresponds to the number of 
independent samples w.r.t. $\mu$ 
which would approximately yield the same mean squared error as the MCMC estimator $S_{n,n_0}(f)$ for computing $\bbE_\mu(f)$.
This can be justified under the assumption that $\vxi_{n_0}\sim\mu$, since then
by virtue of \cite[Proposition~3.26]{Ru12}
we have
\begin{align*}
	\lim_{n\to \infty} n\cdot \mathbb{E}\abs{S_{n,n_0}(f)-\mathbb{E}_\mu(f)}^2
	& = \sigma_{f,M}^2,\\
	1+2\sum_{k\geq 0}\gamma_f(k) & = \frac {\sigma_{f,M}^2}{\mathbb{E}_\mu(f^2)-\mathbb{E}_\mu(f)^2}
\end{align*}
where $\sigma_{f,M}^2$ denotes the asymptotic variance 
of the estimator $S_{n,n_0}(f)$ as in Section \ref{sec:MC}. 

For numerical simulations we use an uniform 
discretization of $[0,1]$ with $\Delta x = 2^{-9}$
and apply the trapezoidal rule for evaluating integrals w.r.t. $\d x$.
Furthermore, we truncate the expansion \eqref{equ:KLE} 
after $N$ terms where we vary 
$N$ in order to test the Metropolis algorithms for dimension independent performance. 
The noise-free 
observations are generated by $u(x) = 2\sin(2\pi x)$.
We also consider different 
noise levels $\sigma_\varepsilon$
to examine the effect of smaller variances $\sigma_\varepsilon^2$, 
leading to more concentrated posterior distributions $\mu$, to the performance of the Metropolis algorithms.
In all cases we take $n_0=10^5$ as burn-in length and $n=10^6$ as sample size.
We use 
$f(\vxi) := \int_0^1 \e^{u(x,\vxi)}\d x$ as the quantity of interest\footnote{
We also studied other functions such as
$f(\vxi) = \xi_1$, $f(\vxi) = \max_x \e^{u(x,\vxi)}$ and $f(\vxi) = p(0.5,\vxi)$
but the results of the comparison were essentially the same.}.
To estimate the $\text{ESS}$ we use the initial monotone sequence estimators\footnote{We also estimated the $\text{ESS}$ by batch means ($100$ batches of size $10^4$) to control our simulations. This lead to similar results.}, 
for details we refer to \cite[Section 3.3]{Geyer1992}.

The results of the simulations are illustrated in Figure~\ref{fig:acrf} and Figure~\ref{fig:ess}.
The former displays the estimated 
autocorrelation functions $\gamma_f$ 
resulting from the four Metropolis algorithms for $N=50$ and $\sigma_\varepsilon=0.1$ in (a), for $N=50$ and $\sigma_\varepsilon=0.01$ in (b),
for $N=400$ and $\sigma_\varepsilon=0.1$ in (c) and for $N=400$ and $\sigma_\varepsilon=0.01$ in (d).
In Figure~\ref{fig:ess} we display the estimated $\text{ESS}$ for varying $\sigma_\varepsilon= 0.1, 0.05, 0.025, 0.01$ with fixed $N = 100$ in (a) and varying $N = 50, 100, 200, 400, 800$ with fixed $\sigma_\varepsilon = 0.1$ in (b).

We see in both figures that the performance of pCN and gpCN is independent of the dimension 
and only GN-RW and gpCN perform robustly w.r.t. the noise variance. 
Thus, the gpCN Metropolis seems to be the only algorithm with both desirable properties.
Intuitively, the variance independent performance might come from the fact that
our choice of $C_\Gamma$ incorporates the noise covariance $\sigma_\varepsilon^2 I$ 
in a way as the posterior covariance might depend on. 
Thus, the smaller $\sigma_\varepsilon$ becomes, i.e., the more pronounced the change from prior to posterior is, the more pronounced is also the adaptation in the proposal covariance by $C_{\Gamma} = (C^{-1} + \sigma_{\varepsilon}^{-2} LL^\top)^{-1}$.
Moreover, the gpCN performs best among the four algorithms also in absolute terms of the $\text{ESS}$.

\begin{figure}[h]
\hfill
\begin{minipage}{0.49\textwidth}
	\centering(a) \\
	\includegraphics[width = \textwidth]{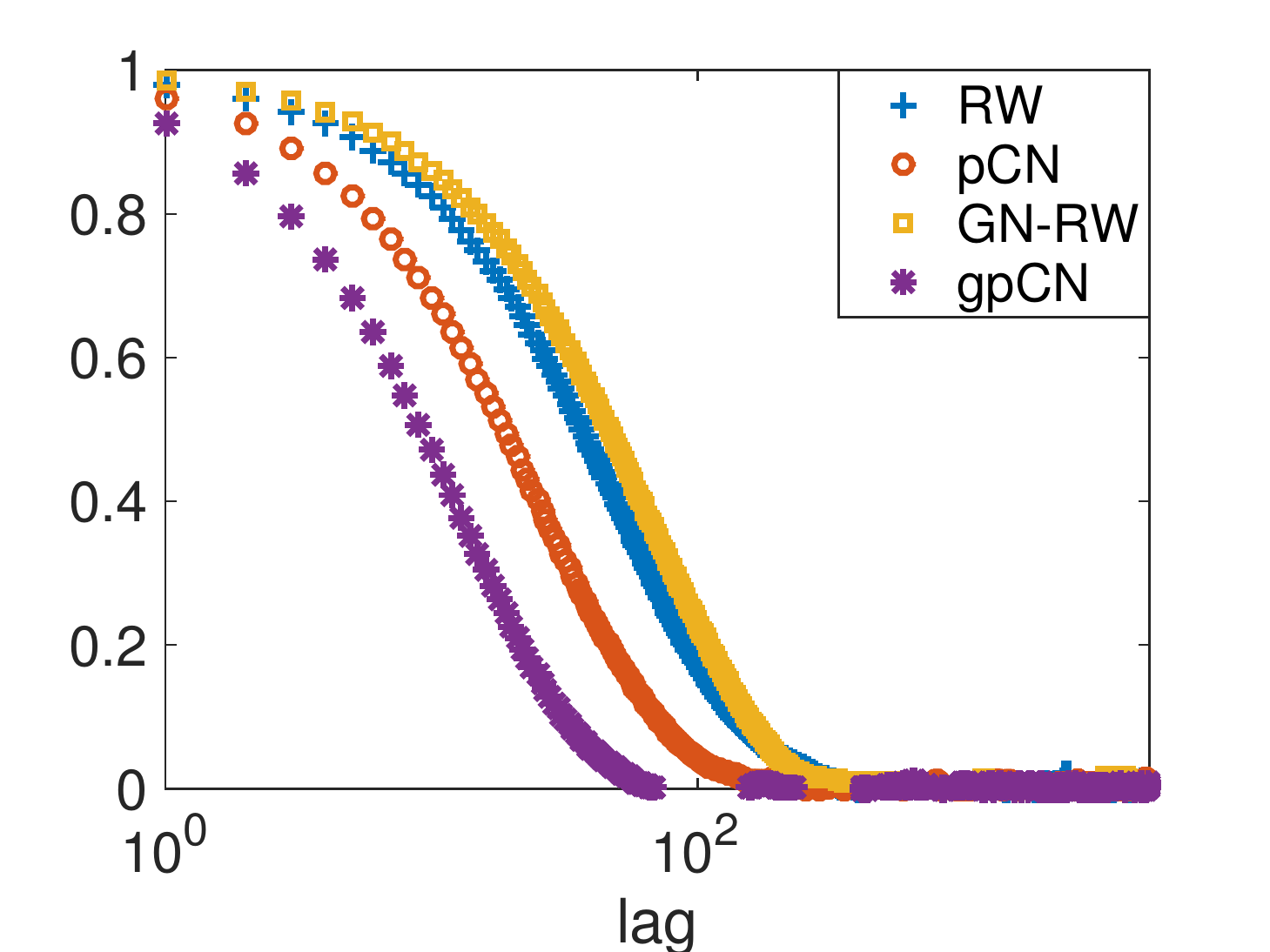}
\end{minipage}
\hfill
\begin{minipage}{0.49\textwidth}
	\centering(b) \\
	\includegraphics[width = \textwidth]{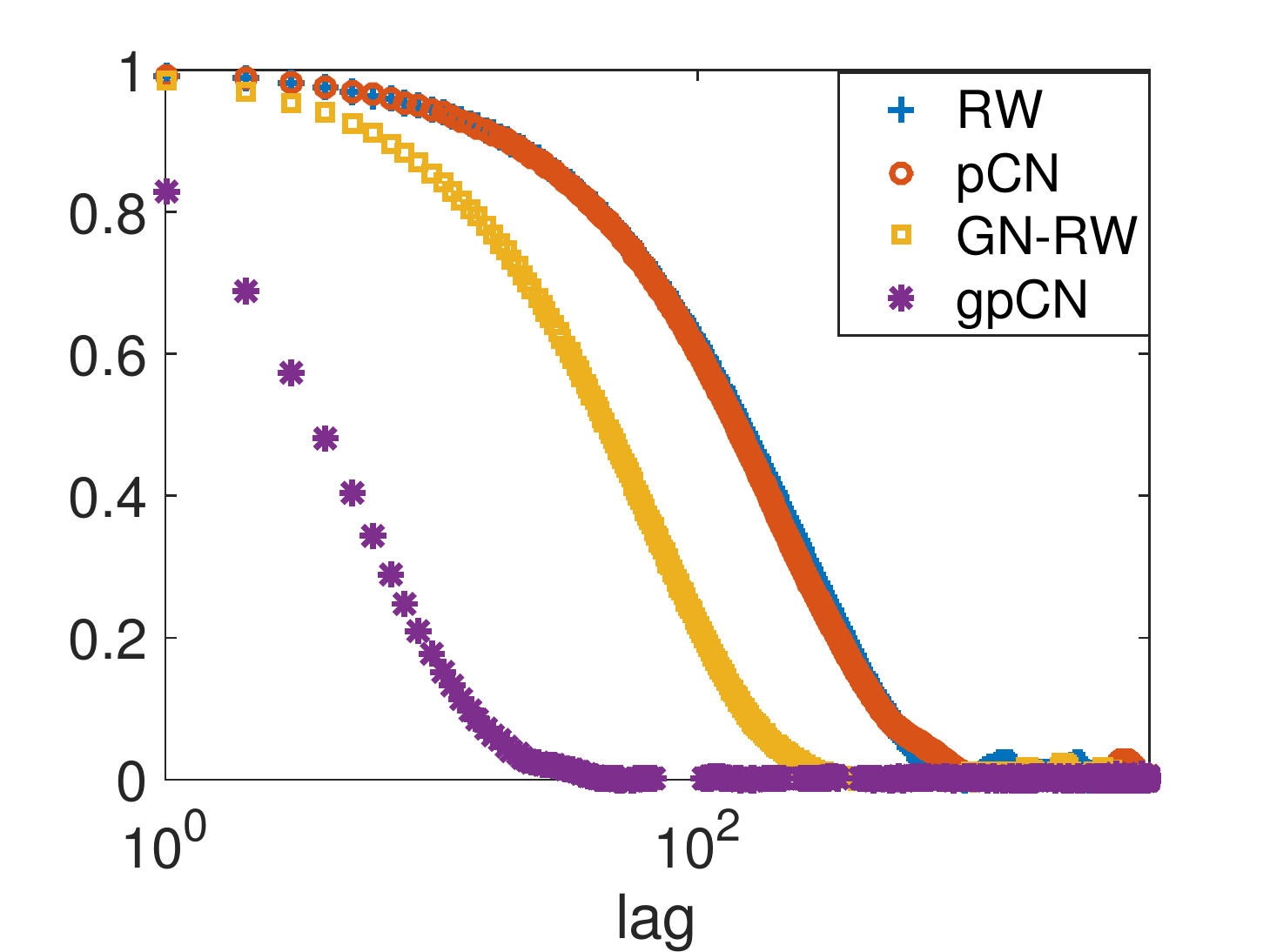}
\end{minipage}
\vspace*{1ex}
\hfill
\begin{minipage}{0.49\textwidth}
	\centering(c) \\
	\includegraphics[width = \textwidth]{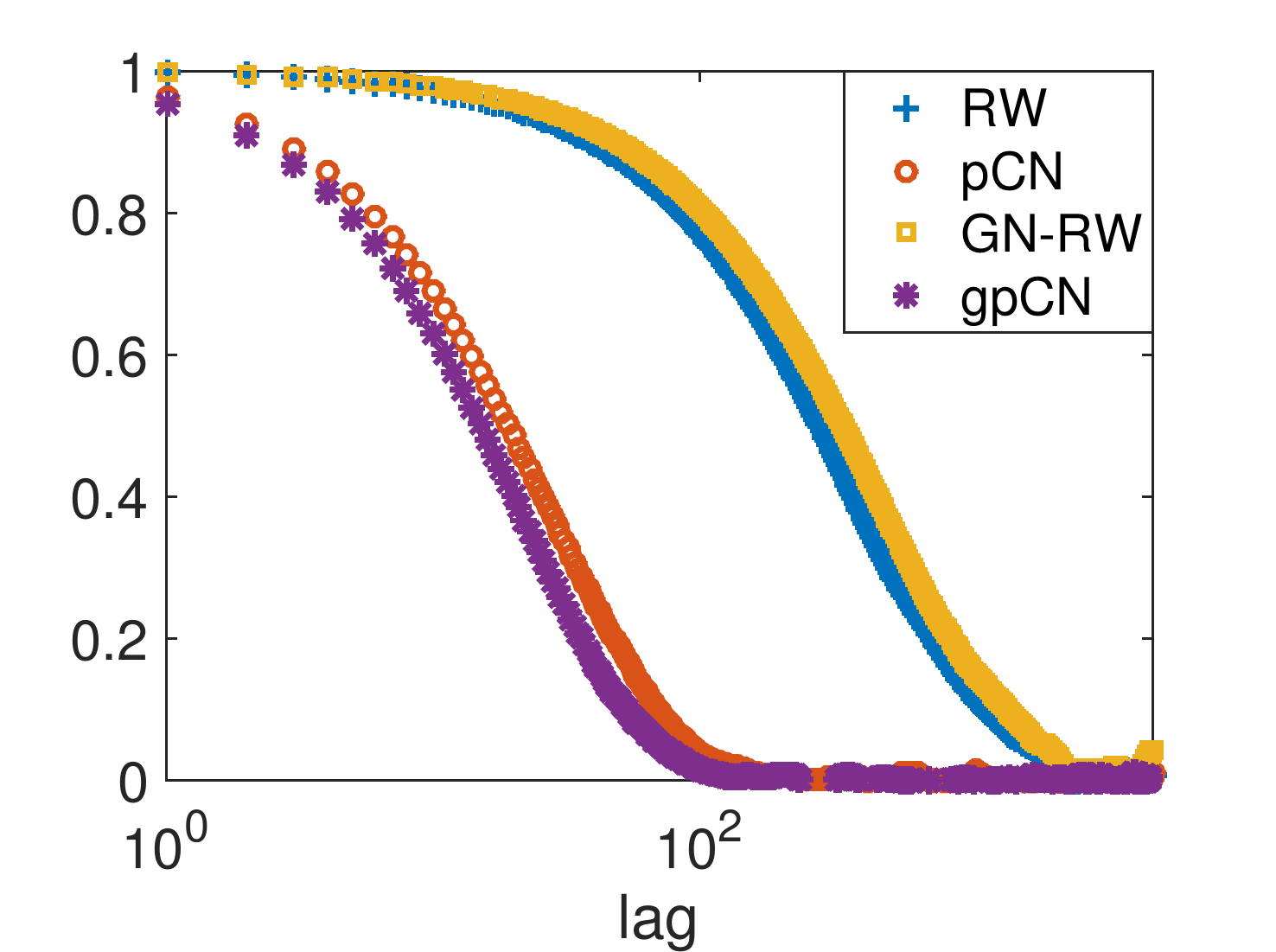}
\end{minipage}
\hfill
\begin{minipage}{0.49\textwidth}
	\centering(d) \\
	\includegraphics[width = \textwidth]{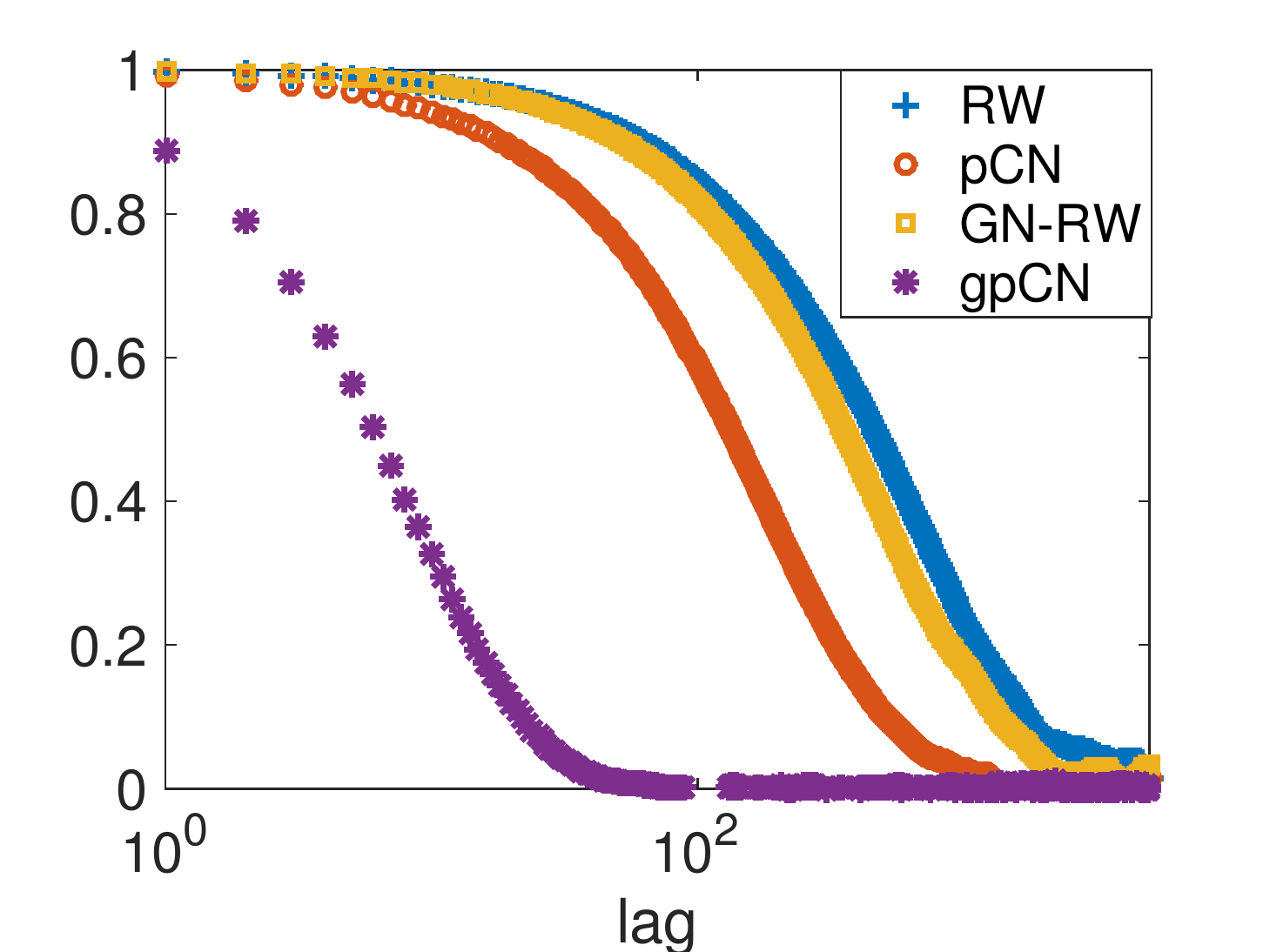}
\end{minipage}
\hfill
\caption{Autocorrelation of $f$ given 
samples generated by the four Metropolis algorithms 
denoted by RW, pCN, GN-RW and gpCN for: 
(a) state dimension $N = 50$ and noise standard deviation $\sigma_\varepsilon = 0.1$; 
(b) $N = 50$ and $\sigma_\varepsilon = 0.01$;
(c) $N = 400$ and $\sigma_\varepsilon = 0.1$;
(d) $N = 400$ and $\sigma_\varepsilon = 0.01$.}
\label{fig:acrf}
\end{figure}

\begin{figure}[h]
\hfill
\begin{minipage}{0.49\textwidth}
	\centering(a) \\
	\includegraphics[width = \textwidth]{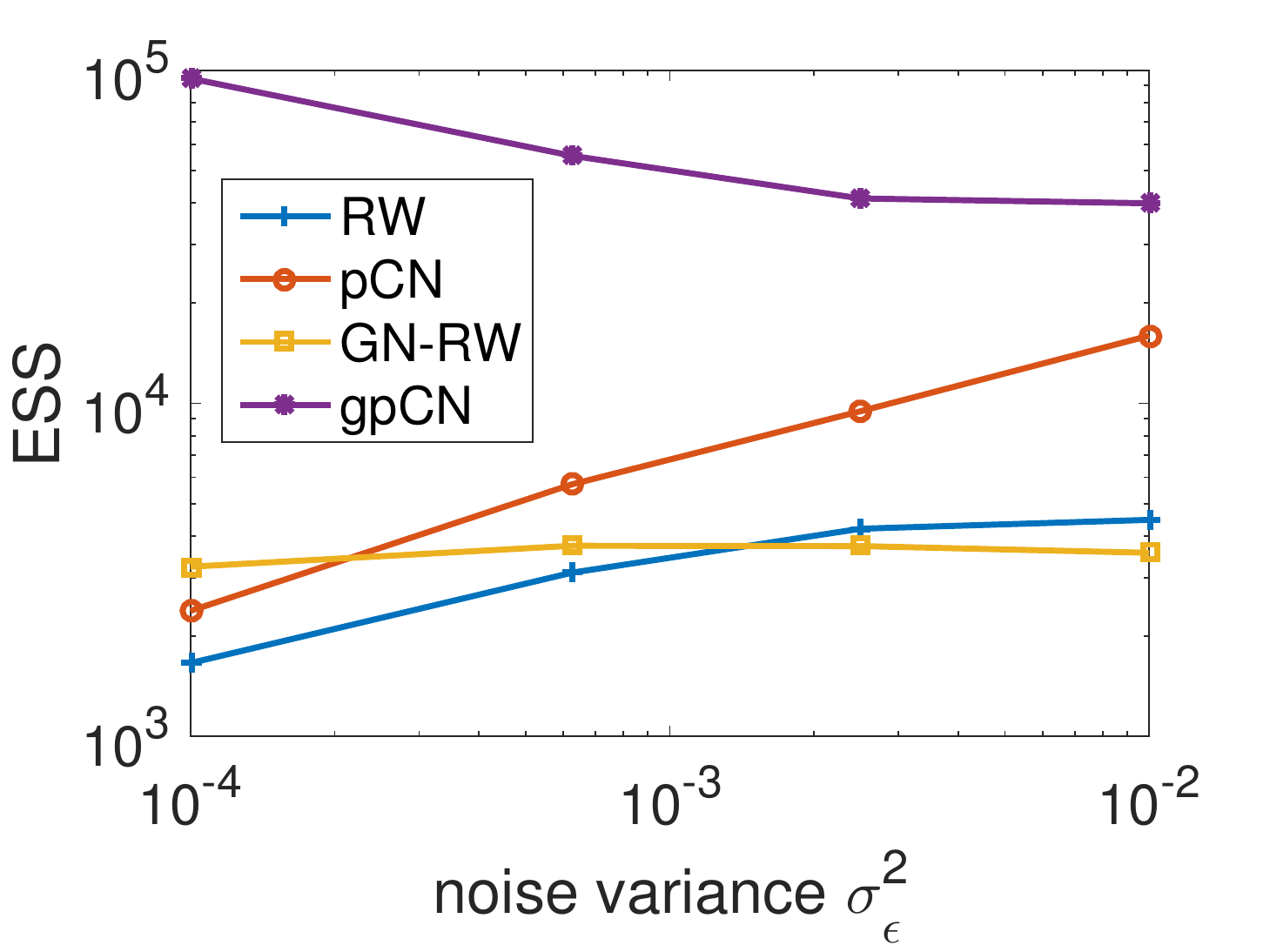}
\end{minipage}
\hfill
\begin{minipage}{0.49\textwidth}
	\centering(b) \\
	\includegraphics[width = \textwidth]{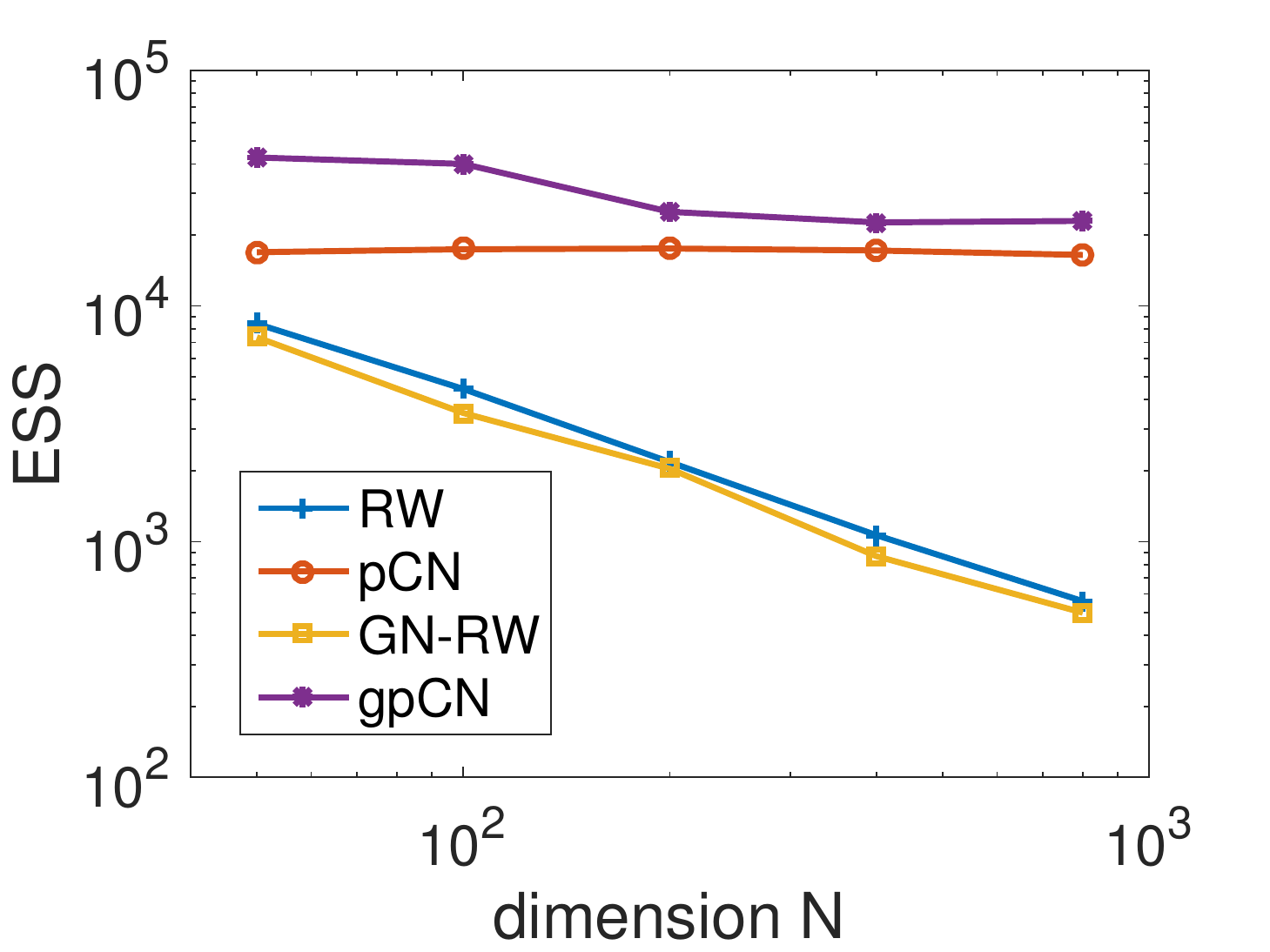}
\end{minipage}
\hfill
\caption{Dependence of empirical $\text{ESS}$ 
for each Metropolis algorithm RW, pCN, GN-RW and gpCN w.r.t.: 
(a) noise variance with fixed state dimension $N=100$; (b) state dimension with fixed noise variance $\sigma^2_\varepsilon = 0.01$.}
\label{fig:ess}
\end{figure}

\section{Qualitative comparison of gpCN Metropolis} \label{sec:Conv}

In this section we develop qualitative comparison arguments 
for Metropolis algorithms in a general setting 
and apply those results to the gpCN Metropolis algorithms.
In particular, we 
relate the existence of a spectral gap 
for the gpCN  to the existence of a spectral gap of the pCN Metropolis.
Here it is worth mentioning 
that in \cite{HaStVo14} sufficient conditions 
for the latter were proven under additional regularity assumptions on the function $\Phi$ in \eqref{equ:mu}. 
With our approach we do not need to rely on those conditions and will benefit from any improvement of the results stated in 
\cite{HaStVo14}.

We start with stating a general comparison result 
for the spectral gaps of Metropolis algorithms with equivalent proposals.
We then 
verify the corresponding assumptions 
for the gpCN Metropolis: positivity and equivalence to the pCN proposal.
In order to derive our main theorem, we consider in Section \ref{sec:restrict} 
restrictions of the target measure $\mu$ to arbitrary $R$-balls in $\mc H$ 
and prove convergence of the gpCN Metropolis to these 
restricted measures.

\subsection{Comparison of spectral gaps} \label{subsec: comp_conduct}
Let $K$ be a $\mu$-reversible transition kernel on $(\mc H,\mc B(\mc H))$, 
i.e., the associated Markov operator $K\colon L_2(\mu) \to L_2(\mu)$ is self-adjoint.
Let the largest element of the spectrum $\spec(K\,|\, L_2^0(\mu))$ be given by
\[
	\Lambda(K) := \sup\{ \lambda \colon \lambda \in \spec(K\,|\, L_2^0(\mu)) \}
\]
and define the \emph{conductance of $K$} (w.r.t. $\mu$) by
\[
	\varphi(K) := \inf_{\mu(A)\in (0,1/2]} \frac{\int_A K(u,A^c)\mu(\d u)}{ \mu(A)}.
\]
Under the assumptions above the Cheeger inequality for Markov operators, 
see \cite{LawlerSokal1988}, given by
\begin{equation} \label{eq:Cheeger}
 	\frac{\varphi(K)^2}{2} \leq 1-\Lambda(K) \leq 2 \varphi(K)
\end{equation}
provides a
useful relation between $\Lambda(K)$ and 
the conductance $\varphi(K)$. 

Let us assume that $M_1$ and $M_2$ are $\mu$-reversible transition kernels 
of Metropolis algorithms with the same acceptance probability $\alpha$ and
proposals $P_1$ and $P_2$, respectively. 
Then, we obtain the following result.

\begin{lem} \label{lem:conductance}
Let $\mu$ be a probability measure on $(\mc H, \mc B(\mc H))$ 
and for $i=1,2$ let
\[
	M_i(u, \d v) =   \alpha(u,v) P_i(u,\d v) +\delta_u(\d v) \,\int_{\mc H} (1 - \alpha(u, w)) \, P_i(u,\d w)
\]
be 
Metropolis kernels. 
Assume that for any $u\in\mc H$ the 
Radon-Nikodym derivative of 
$P_1(u,\d v)$ 
w.r.t. 
 $P_2(u,\d v)$ exists, i.e.,
the proposal kernels admit a density
\[
	\rho(u, v) = \frac{\d P_1(u)}{ \d P_2(u)}(v),\qquad u,v\in\mc H.
\]
If for a number $p> 1$ we have
\begin{equation} \label{eq: kappa}
	\kappa_p := \sup_{\mu(A)\in (0,1/2]} \frac{\int_A \int_{A^c} \rho(u,v)^{p} 
	P_2(u,\d v)\,\mu(\d u)}{\mu(A)} < \infty,
 \end{equation}
then
\[
	\varphi(M_1)  \leq \kappa_p^{1/p}\, \varphi(M_2)^{(p-1)/p}.
\]
\end{lem}
\begin{proof}
 Let $A\in \mc B(\mc H)$ with $\mu(A) \in (0,1/2]$.
 Further, let $q=p/(p-1)$ such that $1/q + 1/p=1$.
 Then
 \begin{align*}
   \int_A M_1(u,A^c) \,\d\mu( u)
	 & = \int_{\mc H} \int_{\mc H} \mathbf{1}_{A^c}(v)\mathbf{1}_{A}(u)\, \alpha(u,v)\, P_1(u,\d v)\, \d\mu( u)\\
	 & = \int_{\mc H} \int_{\mc H} \mathbf{1}_{A^c}(v)\mathbf{1}_{A}(u)\, \alpha(u,v)\,\rho(u;v)\, P_2(u,\d v)\, \d\mu( u).
 \end{align*}
Note that $P_2(u,\d v)\mu(\d u)$ is a probability measure 
on $(\mc H\times \mc H,\mc B(\mc H\times \mc H))$
and we can apply H\"older's inequality according to this measure with parameters $p$ and $q$.
Thus, by using $\alpha(u,v) = \alpha(u,v)^{1/q} \alpha(u,v)^{1/p}$ we obtain
\begin{align*}
	& \int_A M_1(u,A^c)\, \d\mu(u) \\
	&\leq \left( \int_A M_2(u,A^c)\, \d \mu(u)\right)^{1/q} \left( \int_A \int_{A^c} \rho(u,v)^p \alpha(u,v)\, P_2(u,\d v)\,\d\mu(u)\right)^{1/p}\\
 	& \leq \left( \int_A M_2(u,A^c)\, \d \mu(u)\right)^{1/q} \left( \int_A \int_{A^c} \rho(u,v)^{p} P_2(u,\d v)\,\d\mu(u)\right)^{1/p}
\end{align*}
Dividing by $\mu(A)$, applying $\mu(A)^{-1} = \mu(A)^{-1/q}\,\mu(A)^{-1/p}$ and taking the infimum yields
\[
 \varphi(M_1) \leq \varphi(M_2)^{1/q} \kappa_p^{1/p}.
\]
\end{proof}
Employing comparison inequalities in terms of the conductance 
is not an entirely new idea, see for example \cite[Proof of Theorem~4]{LeeLat2014}. 
There the authors obtained a conductance inequality for transition 
kernels with bounded Radon-Nikodym derivatives w.r.t. each other.
An immediate consequence of Lemma \ref{lem:conductance} 
and \eqref{eq:Cheeger} is the following theorem.

\begin{theo}[Spectral gap comparison] \label{theo:comparison_gap}
Let the assumptions of Lemma \ref{lem:conductance} 
be satisfied and let the Markov operators associated with 
$M_1$ and $M_2$ be positive and self-adjoint on $L_2(\mu)$.
Then 
\[
	\left( \frac{\gap(M_1)}{2} \right)^p \leq \kappa_p\, (2\,\gap(M_2))^{(p-1)/2}.
\]
\end{theo}
We apply Theorem \ref{theo:comparison_gap} 
to prove our convergence result for 
the gpCN Metropolis.
We therefore verify in the following section the condition 
that the corresponding Markov operator is positive.

\subsection{Positivity of Metropolis with Gaussian proposals}
Recall that $\langle f,g \rangle_\mu = \int_{\mc H} f g\, \d\mu$ 
denotes the inner-product of $L_2(\mu)$ and that a 
Markov operator $K\colon L_2(\mu) \to L_2(\mu)$ is positive 
if $\langle Kf, f\rangle_{\mu} \geq 0$
for all $f\in L_2(\mu)$.

\begin{lem}[Positivity of proposals] \label{propo:GaussianRW_positiv} 
Let $\mu_0 = N(0,C)$ be a Gaussian measure 
on a separable Hilbert space $\mc H$ and 
let $P(u,\cdot) = N(A u, Q)$ be a $\mu_0$-reversible proposal kernel with a bounded, 
linear operator $A: \mc H\to\mc H$.
If there exists a bounded, linear operator $B: \mc H \to \mc H$ such that
\[
	B^2 = A, \qquad BC = CB^*, 
\]
and $D:=C - BCB^*$ 
is positive and trace class, then, 
the Markov operator associated with the proposal $P$ is positive on $L_2(\mu_0)$.
\end{lem} 
\begin{proof}
Because of the assumptions on $B$ and $D$ we obtain that the proposal kernel $P_1(u,\cdot) = N(Bu,D)$ is well-defined.
Further, since $BCB^* + D=C$ we derive 
\[
 P_1(u,\d v) \mu_0(\d u) 
  = N\left( 
 \begin{bmatrix}  0\\ 0 \end{bmatrix}, 
 \begin{bmatrix}  C & CB^* \\ BC & C\end{bmatrix} \right),
\]
which leads by $BC=CB^*$ to the $\mu_0$-reversibility of $P_1$ and, thus, 
to the self-adjointness of its associated Markov operator in $L_2(\mu_0)$.
It remains to prove that $P_1^2 = P$ holds for 
the associated Markov operators which then immediately yields the assertion.
The equality of the Markov operators is equivalent to the equality of the measures $P_1^2(u,\cdot)$ and $P(u,\cdot)$ for all $u\in \mc H$. 
In order to show that $P_1^2(u,\cdot) = P(u,\cdot)$ for all $u\in\mc H$, we take $(\xi_n)_{n\in\mathbb{N}}$ to be an i.i.d. sequence with $\xi_1\sim N(0,D)$ and construct an auxiliary Markov chain by
\[
	X_{n+1} = B X_{n} + \xi_n, \quad n\geq 1,
\]
where $X_1=u$ for an arbitrary $u\in \mc H$.
The transition kernel of the chain $(X_n)_{n\in\mathbb{N}}$ is the kernel $P_1$.
In particular, for $G\in \mathcal{B}(\mc H)$ holds $\mathbb{P}[X_3 \in G] = P_1^2(u,G)$.
By
\[
 X_3 = B X_2 + \xi_2 = B^2 u + B \xi_1 + \xi_2
\]
and $B\xi_1 + \xi_2 \sim N(0,BDB^*+D)$ we obtain $X_3\sim N(B^2u,BDB^*+D)$.
Due to the assumptions we have $B^2 = A$ and
\[
 BDB^*+D = B(C-BCB^*)B^* + C-BCB^* = C-ACA^*.
\]
The last step $C-ACA^* = Q$ follows by the assumed $\mu_0$-reversibility of $P$, because we know from Section \ref{sec:well_gpCN} that $P$ being $\mu_0$-reversible is equivalent to $A$ and $Q$ satisfying $AC = CA^* $ and $ACA^* + Q = C$.
We thus arrive at $X_3 \sim  N(Au, Q)$ which proves $P_1^2(u,\cdot) = P(u,\cdot)$.
\end{proof}

The next lemma extends the previous result 
to Markov operators associated with Metropolis algorithms. 
The proof follows by the same line of arguments as developed 
in \cite[Section~3.4]{RuUl13} and is therefore omitted.

\begin{lem}[Positivity of Metropolis kernels] \label{propo:MH_positiv} 
Let $\mu$ be a measure on $\mc H$ given by \eqref{equ:mu} and let $P$ be a $\mu_0$-reversible proposal kernel whose associated Markov operator is positive on $L_2(\mu_0)$.
Then the Markov operator associated with  
a $\mu$-reversible Metropolis kernel
\[
 M(u,\d v) = \alpha(u,v) P(u,\d v) 
 + \delta_{u}(\d v) \int_{\mc H} (1-\alpha(u,w))P(u,\d w)
\]
with $\alpha(u,v) = \min\{1,\frac{\d\mu}{\d\mu_0}(v) \frac{\d\mu_0}{\d\mu}(u) \}$
is positive on $L_2(\mu)$.
\end{lem} 

The previous two lemmas lead to the following result about the gpCN Metropolis.
\begin{theo}[Positivity of gpCN Metropolis] \label{theo:gpCN_positiv} 
Let $\mu_0 = N(0,C)$ and $\mu$ as in \eqref{equ:mu} and let $M_{\Gamma}$ denote the gpCN Metropolis kernel as in Corollary \ref{cor:gpCN}. Then the associated Markov operator $M_{\Gamma}$ is self-adjoint and positive on $L_2(\mu)$.
\end{theo} 
\begin{proof}
It is enough to verify the assumptions of Lemma~\ref{propo:GaussianRW_positiv} for the gpCN proposal.
Recall that $P_{\Gamma}(u,\cdot) = N(A_{\Gamma}u, s^2C_\Gamma)$ 
which is $\mu_0$-reversible by construction with bounded 
$A_{\Gamma} = C^{1/2} \sqrt{I-s^2 (I+H_\Gamma)^{-1}} C^{-1/2}$.
By choosing
\[
 B := C^{1/2} \sqrt[4]{I-s^2 (I+H_\Gamma)^{-1}} C^{-1/2},
\]
we obtain $B^2 = A_{\Gamma}$ and $BC = CB^*$.
Moreover,  
\[
 	D= C-BCB^* = C^{1/2} (I-\sqrt{I-s^2(I+H_\Gamma)^{-1}}) C^{1/2}.
\]
The eigenvalues of $I-\sqrt{I-s^2(I+H_\Gamma)^{-1}}$ take the form 
$1 - \sqrt{1-\frac{s^2}{1+\lambda}} \geq 0$ with $\lambda\geq0$ being an eigenvalue of $H_\Gamma$.
Thus, $I-\sqrt{I-s^2(I+H_\Gamma)^{-1}}$ 
is positive and bounded which yields $D$ being positive and trace class since $D$ is then a product of two Hilbert-Schmidt and one bounded operator.
Thus, the conditions of Lemma \ref{propo:GaussianRW_positiv} are satisfied and the assertion follows.
\end{proof}

\subsection{Density between pCN and gpCN proposal} \label{sec:gpCN_density}
In this section we show 
that for any state $u\in \mc H$ the gpCN proposal 
is equivalent to the pCN proposal in the sense of measures.
Moreover, we will also derive an integrability result for the corresponding density.
For proving the equivalence we need the following technical result.

\begin{lem}\label{lem:Delta_Gamma}
Let the assumptions of Corollary \ref{cor:gpCN} be satisfied and define the bounded, linear operator $\Delta_\Gamma:\mc H \to \mc H$ by
\begin{equation}\label{equ:Delta_Gamma}
\Delta_\Gamma := A_{0}-A_{\Gamma} = \sqrt{1-s^2} I - C^{1/2} \; \sqrt{I - s^2 \left( I + H_\Gamma \right)^{-1} }  C^{-1/2}.
\end{equation}
Then $\mathrm{Im}\, \Delta_\Gamma \subseteq \mathrm{Im}\, C^{1/2}$, i.e., $C^{-1/2}\Delta_\Gamma$ is a bounded operator on $\mc H$.
\end{lem}
The proof of this lemma can be found in Appendix \ref{sec:proof_lem:Delta_Gamma}.
It is similar to the proof of Lemma \ref{lem:A} and again rather technical.
However, Lemma \ref{lem:Delta_Gamma} 
ensures that we can apply the Cameron-Martin theorem, 
Theorem~\ref{thm: Cameron_Martin_form} in Appendix \ref{sec:Gaussian}, 
in the proof of the following result.
The other main tool for 
deriving the next theorem is a 
variant of the Feldman-Hajek theorem 
as stated in Theorem \ref{thm: Feldman_form} in Appendix \ref{sec:Gaussian}.

\begin{theo}[Density of pCN w.r.t. gpCN] \label{theo:gpCN_density}
With the notation and assumptions of Corollary \ref{cor:gpCN} holds the following.
\begin{enumerate}
\item
The measures $\mu_0 = N(0,C)$ 
and $\mu_\Gamma = N(0,C_\Gamma)$ are equivalent with
\begin{equation} \label{eq: density_fh}
	\pi_{\Gamma}(v) := \frac {\d \mu_{0}}{\d \mu_{\Gamma}}(v) 
	= \frac {\exp\left( \frac 12 \langle \Gamma v, v \rangle\right)}{\sqrt{\det(I+H_\Gamma)}}.
\end{equation}
\item
For $u\in\mc H$ the measures $P_{0}(u,\cdot)$ 
and $P_{\Gamma}(u,\cdot)$ are equivalent with
\begin{equation} \label{eq: density_fh_cm}
	\frac {\d P_{0}(u)}{\d P_{\Gamma}(u)}(v) 
	= \pi_{\mathrm{CM}}\Big( \Delta_\Gamma u, 
	\frac {1}s(v - A_{\Gamma}u) \Big) \, 
	\pi_{\Gamma}\Big( \frac {1}s(v - A_{\Gamma}u) \Big)
\end{equation}
where $\Delta_\Gamma$ as in \eqref{equ:Delta_Gamma} and 
\begin{equation} \label{eq: density_cm}
	\pi_{\mathrm{CM}}(h,v) := \exp\left( - \frac 12 \|C^{-1/2} h\|^2 + \langle C^{-1} h, v \rangle \right).
\end{equation}
(The subscript in $\pi_{\mathrm{CM}}$ indicates the Cameron-Martin formula.)
\end{enumerate}
\end{theo}
\begin{proof}
We prove \eqref{eq: density_fh} by verifying the assumptions of Theorem \ref{thm: Feldman_form} from Appendix \ref{sec:Gaussian}.
We observe
\[
	I - C^{-1/2} C_\Gamma C^{-1/2} = I - (I+H_\Gamma)^{-1}
\]
and set $T_\Gamma :=I - (I+H_\Gamma)^{-1}$.
The eigenvalues $(t_n)_{n\in\bbN}$ of the self-adjoint operator $T_\Gamma$ are given by
\[
	t_n = 1 - \frac 1{1+\lambda_n}= \frac {\lambda_n}{1+\lambda_n} < 1
\]
where $(\lambda_n)_{n\in\bbN}$ are the eigenvalues of the positive 
trace class operator $H_\Gamma$.
Thus,  $T_\Gamma$ is also trace class and satisfies $\langle T_\Gamma u, u\rangle < \|u\|^2$ for any $u\in\mc H$.
Then, the assertion follows by Theorem \ref{thm: Feldman_form} and 
\[
	T_\Gamma (I-T_\Gamma)^{-1} = \left( I - (I+H_\Gamma)^{-1} \right) (I+H_\Gamma) = H_\Gamma
\]
as well as
\[
	\langle H_\Gamma \, C^{-1/2}v, C^{-1/2}v \rangle = \langle \Gamma v, v \rangle \qquad \forall v \in \mc H.
\]
To show the equivalence of $P_{0}(u,\cdot)$ and $P_{\Gamma}(u,\cdot)$ for any $u\in\mc H$ we introduce the auxiliary kernel $K_{\Gamma}(u,\cdot) = N(A_{\Gamma} u, s^2C)$.
The first assertion and a simple change of variables, see Lemma \ref{lem: change_of_variables} in the appendix, lead to
\[
\frac {\d K_{\Gamma}(u)}{\d P_{\Gamma}(u)}(v) 
= \pi_{\Gamma}\left( \frac 1s \left[v - A_{\Gamma}u\right] \right),  \qquad u,v\in\mc H.
\]
Thus, it remains to prove the equivalence of $K_{\Gamma}(u,\cdot)$ and $P_{0}(u,\cdot)$ for any $u\in \mc H$.
By the Cameron-Martin formula, see Theorem \ref{thm: Cameron_Martin_form} in Appendix \ref{sec:Gaussian}, 
this holds iff 
\[
	\Image(A_{\Gamma} - \sqrt{1-s^2} I) \subseteq \Image(C^{1/2})
\]
which was shown in Lemma \ref{lem:Delta_Gamma}.
Now Theorem \ref{thm: Cameron_Martin_form} 
combined with a change of variables, see Lemma \ref{lem: change_of_variables}, then yields
\[
	\frac{\d P_{0}(u)}{\d K_{\Gamma}(u)}(v) 
	= \pi_{\mathrm{CM}}\left( [ \sqrt{1-s^2} I-A_{\Gamma}] u, \frac{1}{s}(v - A_{\Gamma}u) \right)
\]
and the assertion follows by
\[
	\frac{\d P_{0}(u)}{\d P_{\Gamma}(u)}(v) 
	=
	\frac{\d P_{0}(u)}{\d K_{\Gamma}(u)}(v)
	\frac{\d K_{\Gamma}(u)}{\d P_{\Gamma}(u)}(v).
	\qedhere
\]
\end{proof}

Note that Theorem \ref{theo:gpCN_density} implies that for any $\Gamma_1, \Gamma_2\in\mc L_+(\mc H)$ there exists a density between the two gpCN proposals $P_{\Gamma_1}(u)$ and $P_{\Gamma_2}(u)$.
However, for the application
of Theorem \ref{theo:comparison_gap} we
still have to verify condition \eqref{eq: kappa}.
This is partly addressed in the following result.

\begin{theo}[Integrability of gpCN density] \label{theo:int_rho}
Let the assumptions of Lemma \ref{lem:Delta_Gamma} be satisfied and set
\[
	\rho_{\Gamma}(u, v) := \frac{\d P_{0}(u)}{\d P_{\Gamma}(u)}(v), \qquad u,v\in\mc H.
\]
Then, for any $0 < p < 1 + \frac {1}{2\|H_\Gamma\|}$ 
there exist constants $c = c(p, H_\Gamma) < \infty$ 
and $b = b(p, \|C^{-1/2}\Delta_\Gamma\|) < \infty$ such that
\[
	\int_{\mc H} \rho^{p}_{\Gamma}(u, v) \, P_{\Gamma}(u, \d v) 
	\leq c \, \exp\left(\frac{b}{2} \|u\|^2\right).
\]
\end{theo}
\begin{proof}
We employ the same notation as in Theorem \ref{theo:gpCN_density}, i.e., 
let $\mu_0 = N(0,C)$ 
and $\mu_\Gamma = N(0,C_\Gamma)$ as well as $\pi_{\Gamma}$ and $\pi_{\mathrm{CM}}$ be as in \eqref{eq: density_fh} and \eqref{eq: density_cm}, respectively.
By Theorem \ref{theo:gpCN_density} we know
\[
	\rho_{\Gamma}(u, v) = \pi_{\mathrm{CM}}\Big( \Delta_\Gamma u, \frac {1}s(v - A_{\Gamma}u) \Big)\; 
	\pi_{\Gamma}\Big( \frac {1}s(v - A_{\Gamma}u) \Big).
\]
By first applying a change of variables, see Lemma \ref{lem: change_of_variables}, and then the 
Cauchy-Schwarz inequality we obtain
\begin{align*}
 	\int_{\mc H} \rho_{\Gamma}^{p}(u, v) \, P_{\Gamma}(u, \d v)
 	& = \int_{\mc H} \pi^{p}_{\mathrm{CM}}(\Delta_\Gamma u, v) \; \pi^{p}_{\Gamma}(v) \, \mu_\Gamma(\d v)\\
 	& = \int_{\mc H} \pi^{p}_{\mathrm{CM}}(\Delta_\Gamma u, v) \; \pi^{p-1}_{\Gamma}(v) \, \mu_0(\d v)\\
	& \leq \left( \int_{\mc H} \pi^{2p}_{\mathrm{CM}}(\Delta_\Gamma u, v) \mu_0(\d v) \right)^{1/2} \; 
	\left( \int_{\mc H} \pi^{2p-2}_{\Gamma}(v) \mu_0(\d v) \right)^{1/2}.
\end{align*}
Furthermore, we have by applying \eqref{equ:int_Wu} from Appendix \ref{sec:Gaussian}
\begin{align*}
	\int_{\mc H} \pi^{2p}_{\mathrm{CM}}\Big( \Delta_\Gamma u, v\Big) \, \mu_0(\d v)
	& = \int_{\mc H} \e^{- \frac {2p}2 \|C^{-1/2}\Delta_\Gamma u \|^2} \e^{2p\; \langle C^{-1} \Delta_\Gamma u, v\rangle} \, \mu_0(\d v)\\
	& = \exp\left((2p^2-p) \|C^{-1/2} \Delta_\Gamma u\|^2\right).
\end{align*}
We apply
$\|C^{-1/2} \Delta_\Gamma u\| \leq \|C^{-1/2} \Delta_\Gamma\|\, \|u\|$ 
and set 
\[
	b := (2p^2-p)\,\|C^{-1/2} \Delta_\Gamma\|.
\]
Note, that $b\leq 0$ for $p\leq \frac 12$.
Due to the assumptions on $p$ we have
\[
	\langle (2p-2)H_\Gamma v, v\rangle < \frac {\langle H_\Gamma v, v\rangle}{\|H_\Gamma\|} \leq \|v\|^2, \qquad v\in \mc H.
\]
Thus, we can apply \eqref{equ:int_Tu} from Appendix \ref{sec:Gaussian} and get
\begin{align*}
	\int_{\mc H} \pi^{2p-2}_{\Gamma}(v) \mu_0(\d v)
	& = \int_{\mc H} \frac {\exp\left( \frac 12 \langle (2p-2)H_\Gamma \, C^{-1/2}v, C^{-1/2}v \rangle\right)}{\det(I+H_\Gamma)^{(2p-2)/2}} \, \mu_0(\d v)\\
	& = \left( \det(I-(2p-2)H_\Gamma) \; \det(I+H_\Gamma)^{2p-2}\right)^{-1/2}\\
	& =: c^2.
\end{align*}
Since $H_\Gamma$ is positive and trace class, $\det(I+H_\Gamma)$ 
is well-defined (see Appendix \ref{sec:Gaussian}) and $\det(I+H_\Gamma) \in [1,\infty)$. 
Furthermore, due to $\langle (2p-2)H_\Gamma v, v\rangle < \|v\|^2$, 
the eigenvalues of $(2p-2)H_\Gamma$ lie within $[0,1)$ 
which ensures that $\det(I-(2p-2)H_\Gamma)>0$ and, hence $0 < c^2<\infty$.
This proves the assertion.
\end{proof}
Thus, the above theorem allows 
us to estimate the integral in \eqref{eq: kappa}. 
We obtain for $0 < p < 1 + 1/(2\|H_\Gamma\|)$ that
\begin{align*}
	\int_A \int_{A^c} \rho_{\Gamma}(u;v)^{p} 
	P_{\Gamma}(u,\d v)\,\mu(\d u) 
	\leq c\, \int_A \exp\left(\frac {b}2\, \|u\|^2\right) \mu(\d u).
\end{align*}
Unfortunately, if we divide the right-hand side by $\mu(A)$ 
and take the supremum over all $\{A: 0 < \mu(A) \leq 0.5\}$ this is 
unbounded.
In the next section we introduce restrictions of the target measure for which we 
can circumvent this problem.

\subsection{Restrictions of the target measure} \label{sec:restrict}

In order to show boundedness of $\kappa_p$ from \eqref{eq: kappa} 
for the gpCN proposal we consider restrictions of the target measure to bounded sets.
For appropriately chosen sets, the restricted measures become arbitrarily close to the target measure.
Let $R\in (0,\infty]$ and set
\[
 \mc H_R := \{ u\in \mc H \colon \norm{u}{}< R \}.
\]
\begin{defi}[Restricted measure] \label{defi:mu_R}
Let $\mu$ be a probability measure on $(\mc H, \mc B(\mc H))$ and $R\in (0,\infty]$. 
We define its restriction to $\mc H_R$ as the probability measure $\mu_R$ on $\mc H$ given by
\begin{equation}\label{equ:mu_R}
	\mu_R(\d u) := \frac{1}{\mu(\mc H_R)} \mathbf{1}_{\mc H_R}(u) \mu(\d u).
\end{equation}
\end{defi}
For sufficiently large $R$ the measure $\mu_R$ is close to $\mu$, because
\[
	\| \mu_R - \mu \|_{\text{tv}} 
	= \int_{\mc H} \left|\frac{\d \mu_R}{\d \mu}(u) - 1 \right| \d\mu( u) 
	= \mu(\mc H_R^c) + 1 - \mu(\mc H_R)
	= 2\mu(\mc H_R^c)
\]
and since $\mu$ is a probability 
measure on $(\mc H,\mathcal{B}(\mc H))$ 
there exists for any $\varepsilon>0$ a number $R>0$ 
such that $2\mu(\mc H_R^c)< \varepsilon$.
Let us mention here that restricted measures appear, 
for example, also in \cite[Equation (3.5)]{Bou-Rabee13}
and in the recent work \cite{HuYaLi2015},
in order to analyze the convergence of Metropolis-Hastings based algorithms.

We ask now whether good convergence properties of a $\mu$-reversible transition kernel $K$ are inherited on a 
suitably modified $\mu_R$-reversible
transition kernel $K_R$.
\begin{defi}[Restricted transition kernel] \label{defi:K_R}
Let $K$ be a transition kernel on $\mc H$ and $R\in (0,\infty]$. 
We define its restriction to $\mc H_R$ as the following transition kernel 
$K_R\colon \mc H \times \mc B(\mc H) \to [0,1]$
given by
\begin{equation}\label{equ:K_R}
	K_R(u,\d v) := \mathbf{1}_{\mc H_R}(v)\, K(u, \d v) + K(u, \mc H_R^c) \, \delta_u(\d v).
\end{equation}
\end{defi}
Note that if $K$ is $\mu$-reversible, then $K_R$ is $\mu_R$-reversible and if $K$ is of Metropolis form  \eqref{eq: metro_kern}, then so is $K_R$.

\begin{propo} \label{propo: K_R_reversible}
Let $\mu$ be a probability measure on $(\mc H, \mc B(\mc H))$ and $K$ be a $\mu$-reversible transition kernel.
Then for any $R>0$ the transition kernel $K_R$ given in \eqref{equ:K_R} is 
$\mu_R$-reversible with $\mu_R$ as in \eqref{equ:mu_R}.
Moreover, for a Metropolis kernel $M$ of the form \eqref{eq: metro_kern} 
the corresponding restricted kernel $M_R$ is again a Metropolis kernel
\[
	M_R(u, \d v) =  \alpha_R(u,v)P(u,\d v) + \delta_u(\d v)\left(1-\int_{\mc H} \alpha_R(u,w)P(u,\d w)\right)
\]
with $\alpha_R(u,v) := \mathbf{1}_{\mc H _R}(v) \alpha(u,v)$.
\end{propo}
\begin{proof}
Recall that $K$ is $\mu$-reversible iff
\[
 \int_A K(u,B)\,\d\mu(u) = \int_B K(u,A)\, \d \mu(u), \qquad \forall A,B\in \mathcal{B}(\mc H).
\]
Let $A,B \in \mc B(\mc H)$. We have
\begin{align*}
& \int_A K_R(u,B) \, \d\mu_R(u)
 = \int_A K(u,B\cap \mc H_R)\,\d\mu_R(u)
    + \int_{A\cap B} K(u,\mc H_R^c)\, \d\mu_R(u)\\
& \qquad =
\frac{1}{\mu(\mc H_R)} \int_{A\cap \mc H_R} K(u,B\cap \mc H_R)\, \d\mu(u)
+\int_{A\cap B} K(u,\mc H_R^c)\,\d \mu_R(u).
\end{align*}
Because of the $\mu$-reversibility of $K$ we can interchange $A$ and $B$ which leads to the first assertion.
The second statement follows by
\begin{align*}
 M_R(u,\d v) 
 &  = \mathbf{1}_{\mc H_R}(v) M(u,\d v) + \delta_u(\d v) M(u,\mc H_R^c) \\
 &  = \mathbf{1}_{\mc H_R}(v) \alpha(u,v) P(u,\d v) \\
 & \qquad + \delta_u(\d v)\left(1-\int_{\mc H} \alpha(u,w)P(u,\d w) + \int_{\mc H_R^c} \alpha(u,w)P(u,\d w) \right)\\
 &  = \mathbf{1}_{\mc H_R}(v) \alpha(u,v)P(u,\d v) +
 \delta_u(\d v)\left(1-\int_{\mc H_R} \alpha(u,w)P(u,\d w)\right).
\end{align*}
\end{proof}

Now we ask whether a spectral gap of $K$ on $L_2(\mu)$ implies a spectral gap 
of the Markov operator associated with  
$K_R$ on $L_2(\mu_R)$.
Note that
\[
	K_R f(u) = \int_{\mc H} f(v)\, K_R(u,\d v)
	= \int_{\mc H_R} f(v)\, K(u,\d v) + f(u)\,K(u,\mc H_R^c).
\]
We have the following relation between $\|K_R \|_{\mu_R}$ and $\| K \|_{\mu}$.

\begin{lem} \label{lem:K_R_gap}
With the notation and assumptions from above holds
\begin{equation}  \label{eq: rel_spec_gap}
	\| K_R \|_{\mu_R} 
	\leq \| K \|_{\mu}  + \sup_{u\in \mc H_R} K(u,\mc H_R^c). 
\end{equation}
Furthermore, if the Markov operator $K$ is positive on $L_2(\mu)$, then $K_R$ is also positive on $L_2(\mu_R)$.
\end{lem}
\begin{proof}
For $f \in L_2(\mu_R)$ let
\[
	 (Ef) (u) := \mathbf 1_{\mc H_R}(u) f(u) \in L_2(\mu).
\]
Note that 
$\|f\|_{2,\mu_R} =\frac{1}{\sqrt{\mu(\mc H_R)}}\, \|Ef\|_{2,\mu}$ and for 
$\int_{\mc H_R} f \, \d \mu_R =0$ follows $\int_{\mc H} Ef \, \d \mu = 0$.
Further, for any $f\in L_2(\mu_R)$ we have 
 \begin{align*}
  \| K_R f \|^2_{2,\mu_R} 
  & = \int_{\mc H_R} \left| \int_{ \mc H_R}  f(v)\, K(u,\d v) 
  + f(u)\,K(u, \mc H_R^c)\right|^2 \d\mu_R(u)\\
  & = \int_{\mc H_R} \left| \int_{\mc H} Ef(v)\, K(u,\d v) +  Ef(u) \,K(u,\mc H_R^c) \right |^2 \d\mu_R(u)\\
  & =  \| K(Ef) + g\, Ef \|^2_{2,\mu_R} 
 \end{align*}
with $g(u) := \mathbf 1_{\mc H_R}(u) \, K(u,\mc H_R^c)$. 
Then
\begin{align*}
	\frac{\| K_R f \|_{2,\mu_R}}{\| f \|_{2,\mu_R}}
	& = \frac{ \| K(Ef) + g\,Ef \|_{2,\mu_R}}{\| Ef \|_{2,\mu_R}} 
	= \frac{ \| E(K(Ef)) + g\,Ef \|_{2,\mu}}{\| Ef \|_{2,\mu}}\\
	& \leq \frac{ \| K(Ef)  \|_{2,\mu} + \| g\,Ef \|_{2,\mu}}{\| Ef \|_{2,\mu}}\\
	& \leq \frac{ \| K(Ef)  \|_{2,\mu}}{\| Ef \|_{2,\mu}} + \sup_{u\in \mc H_R} K(u,\mc H_R^c),
\end{align*}
where we applied $\|Ef\|_{2,\mu} \leq \|f\|_{2,\mu}$ in the first inequality.
By taking the supremum over all $f \in L^0_2(\mu_R)$ and because of $E(L^0_2(\mu_R)) \subseteq L^0_2(\mu)$ the first assertion follows.
Moreover, we have for $f \in L_2(\mu_R)$ that
\begin{align*}
	\langle K_R f, f\rangle_{\mu_R}
	& = \int_{\mc H} K_R f(u) \, f(u) \, \mu_R(\d u)\\
	& = \int_{\mc H} \left( \int_{\mc H_R} f(v)\, K(u,\d v) + f(u)\,K(u,\mc H_R^c)\right) \, f(u) \, \mu_R(\d u)\\
	& = \int_{\mc H} \int_{\mc H} (Ef)(v)\, K(u,\d v) \; (Ef)(u) \, \frac{\mu(\d u)}{\mu(\mc H_R)}\\
	& \qquad  + \int_{\mc H} f^2(u)\,K(u,\mc H_R^c) \, \mu_R(\d u).
\end{align*}
The second term is always positive since $f^2(u)\,K(u,\mc H_R^c) \geq 0$ for all $u\in\mc H$ and the first term coincides with $\langle K(Ef), Ef\rangle_{\mu} \, /\mu(\mc H_R)$.
Thus, the second statement is proven.
\end{proof}

Lemma \ref{lem:K_R_gap} tells us that there exists an absolute spectral gap of $K_R$ if there exists an absolute spectral gap of $K$ and $\sup_{u\in \mc H_R} K(u,\mc H_R^c)$ is sufficiently small. 
Indeed, we can apply this result to the pCN Metropolis algorithm.

\begin{theo}[Spectral gap of restricted pCN Metropolis] \label{theo:Gap_restricted}
Let $\mu$ be as in \eqref{equ:mu} and let $M_{0}$ denote the $\mu$-reversible pCN Metropolis kernel.
If there exists a spectral gap of $M_{0}$ in $L_2(\mu)$, 
then for any $\varepsilon>0$ there exists a number $R\in(0,\infty)$ such that $M_{0,R}$ possesses a spectral gap in $L_2(\mu_R)$, i.e.,
\[
	{\rm gap}(M_{0,R}) =  1 - \| M_{0,R} \|_{\mu_R} \geq \gap(M_{0}) - \varepsilon,
\] 
where $\mu_R$ as in \eqref{equ:mu_R} and $M_{0,R}$ according to Definition \ref{defi:K_R}. 
\end{theo}
\begin{proof}
Given the results of Proposition \ref{propo: K_R_reversible} and Lemma \ref{lem:K_R_gap} it suffices to prove that for any $\varepsilon>0$ there exists an $R>0$ such that $\sup_{u\in \mc H_R} M_{0}(u,\mc H_R^c) \leq \varepsilon$.
We recall that the proposal kernel of $M_{0}$ is $P_{0}(u,\cdot) = N(\sqrt{1-s^2}u, s^2 C)$ 
and obtain with $\mu^s := N(0, s^2C)$ that
\begin{align*}
	\sup_{u\in \mc H_R} M_{0}(u,\mc H_R^c) 
	& \leq \sup_{u\in \mc H_R} P_{0}(u,\mc H_R^c)
	= \sup_{u\in \mc H_R} \int_{\| \sqrt{1-s^2} u + v\| \geq R} \d \mu^s(v)\\
	& \leq \sup_{u\in \mc H_R} \int_{\| \sqrt{1-s^2} u\| + \|v\| \geq R}  \d\mu^s(v)\\
	& = \sup_{u\in \mc H_R} \int_{\|v\| \geq R - \sqrt{1-s^2}\|u\|}  \d\mu^s(v)\\
	& \leq \int_{\|v\| \geq (1 - \sqrt{1-s^2}) R}  \d\mu^s(v) 
	= \mu_0( \mc H_{R_s}^c)
\end{align*}
where $R_s = \frac{1 - \sqrt{1-s^2}}{s}R$ and $\mu_0 = N(0, C)$.
Again, since $\mu_0$ is a probability measure on $\mc H$ we know that there exists a number $R$, such 
that $\mu_0( \mc H_{R_s}^c)\leq \varepsilon$.
\end{proof}

\subsection{Spectral gap of restricted gpCN Metropolis}
Now, we are able to formulate and to prove our main convergence result.

\begin{theo}[Convergence of restricted gpCN Metropolis] \label{theo:gpCN_conv}
Let $\mu$ be as in \eqref{equ:mu} and assume that the pCN Metropolis kernel possesses a spectral gap in $L_2(\mu)$, i.e., $\gap(M_{0})>0$.
Then, 
for any $\Gamma \in \mc L_+(\mc H)$ and any $\varepsilon \in (0,\gap(M_{0}))$ 
there exists a number $R_0 = R_0(\varepsilon)\in(0,\infty)$ such that for any $R\geq R_0$ holds
\[
	\norm{\mu-\mu_R}{\text{tv}} < \varepsilon
	\quad \text{and} \quad
	\gap(M_{\Gamma,R})>0
\]
where $\gap(M_{\Gamma,R}) = 1 - \| M_{\Gamma,R} \|_{\mu_R}$ denotes the spectral gap of $M_{\Gamma,R}$ in $L_2(\mu_R)$.
\end{theo}
\begin{proof}
By Theorem~\ref{theo:Gap_restricted} we have that for any $\varepsilon \in (0,\gap(M_0))$ there exists a number $R_0\in (0,\infty)$ such that for any $R\geq R_0$ holds
\[
	\norm{\mu-\mu_R}{\text{tv}} \leq \varepsilon
	\quad \text{and} \quad
	\gap(M_{0,R}) >0.
\]
Moreover, Proposition~\ref{propo: K_R_reversible}, Theorem~\ref{theo:Gap_restricted} 
and Theorem~\ref{theo:gpCN_positiv} yield 
that for any $\Gamma \in  \mc L_+(\mc H)$ the Markov operator associated 
to $M_{\Gamma,R}$ is self-adjoint and positive on $L_2(\mu_R)$.
In particular, $M_{\Gamma,R}$ 
is again a Metropolis kernel with proposal $P_\Gamma$ and acceptance probability $\alpha_R$.
Thus, in order to apply Theorem \ref{theo:comparison_gap} to $M_{0,R}$ and $M_{\Gamma,R}$ 
it remains to verify that there exists a $p>1$ so that
\[
  \kappa_{p,R} := \sup_{\mu_R(A)\in (0,1/2]} \frac{\int_A \int_{A^{c}} 
  \rho_{\Gamma}(u,v)^p\,P_{\Gamma}(u,\d v)\, \d\mu_R(u)}{\mu_R(A)} < \infty
\]
where $\rho_{\Gamma}(u, v) = \frac{\d P_{0}(u)}{\d P_{\Gamma}(u)}(v)$.
By Theorem \ref{theo:int_rho} we have for any $p < 1 + \frac {1}{2\|H_\Gamma\|}$ that
\begin{align*}
	\kappa_{p,R} 
	& \leq  \sup_{\mu_R(A)\in (0,1/2]} \frac{\int_A c \exp\left(\frac {b}2\, \|u\|^2\right) \d\mu_R(u)}{\mu_R(A)} 
	\leq c \exp\left(\frac {b}2\, R^2\right) <\infty.
\end{align*}
Hence, Theorem~\ref{theo:comparison_gap} leads to
\[
	\gap(M_{\Gamma,R})^{(p-1)/2} \geq \frac{1}{2^{(3p-1)/2}} \; \frac{\gap(M_{0,R})^{p}}{\kappa_{p,R}} > 0
\]
which proves the assertion.
\end{proof}

Theorem~\ref{theo:gpCN_conv} tells us that the 
corresponding restricted gpCN Metropolis converges exponentially 
fast to any, arbitrarily close, restriction $\mu_R$ of $\mu$ whenever 
the pCN Metropolis has a spectral gap, e.g., under the conditions of \cite[Theorem~2.14]{HaStVo14}.
In particular, Theorem~\ref{theo:gpCN_conv} is a statement about the inheritance of geometric convergence 
from the pCN to the restricted gpCN Metropolis.
We emphasize that a quantitative comparison of their spectral gaps is
not proven.
We provide a lower bound for the spectral gap of $\gap(M_{\Gamma,R})$ in nonlinear terms of the spectral gap of the pCN Metropolis.
Additionally, the stated estimate behaves rather poor in $R$, more precise, it decays exponentially as $R\to \infty$.

Although we argued in the above theorem with restrictions of $\mu$ 
in order to bound $\kappa_p$ from Theorem~\ref{theo:comparison_gap}, 
let us mention that, in simulations when $R$ is sufficiently large
one cannot distinguish between $\mu$ and $\mu_R$ as well as between Markov
chains with transition kernels $M_\Gamma$ and $M_{\Gamma,R}$.

Moreover, 
we conjecture that the gpCN Metropolis targeting $\mu$ has a strictly positive
spectral gap whenever the pCN Metropolis has one. 
Recalling the results of the numerical simulations 
in Section~\ref{subsec: numerics} we even conjecture that 
the spectral gap of the gpCN Metropolis with suitably chosen $\Gamma \in \mc L_+(\mc H)$ is much larger than the one of the pCN Metropolis.

\section{Outlook on gpCN proposals with state-dependent covariances} \label{subsec: loc_gpCN}
In this section we comment on state-dependent proposal covariances as they are a natural extension of the idea behind the gpCN proposal.
The advantage of such a state-dependent approach is that
the resulting Metropolis algorithm might be even better adapted to the
target measure by allowing locally different proposal covariances.
For an illustrative motivation of state-dependent proposal 
covariances we refer to \cite{GirolamiCalderhead2011},\cite{MartinEtAl2012}
and for recent positive and negative theoretical results we refer to \cite{Livingstone15}. 
In the Hilbert space setting we are now able to define MH algorithms 
by means of Theorem \ref{theo:gpCN_density}.
Consider the proposal kernel
\begin{equation} \label{equ:Ploc}
	P_\mathrm{loc}(u,\cdot) = N(A_{\Gamma(u)} u, s^2 C_{\Gamma(u)})
\end{equation}
where we assume 
that for $u\in \mc H$ we have $\Gamma(u)\in\mc L_+(\mc H)$
and that the corresponding mapping $u \mapsto \Gamma(u) $
is measurable. 
Further, 
by $A_{\Gamma(u)}$ and $C_{\Gamma(u)}$ we denote
the components of the gpCN proposal for $\Gamma = \Gamma(u)$.
Following the heuristic presented 
in Section \ref{subsec: Motiv} for Bayesian inference problems 
where $\Phi$ in \eqref{equ:mu} is of the form \eqref{equ:Phi}, 
we could chose for instance
\begin{equation} \label{eq: local_G}
 	\Gamma(u) = \nabla G(u)^*\, \Sigma^{-1}\, \nabla G(u).
\end{equation}
When considering the measure 
$\eta_\mathrm{loc}(\d u , \d v) = P_\mathrm{loc}(u,\d v)\mu_0(\d u)$ 
we notice that $\eta_\mathrm{loc}$ is no longer a Gaussian measure due to the dependence of $\Gamma$ on $u$.
However, to construct a $\mu$-reversible Metropolis kernel 
with the proposal $P_\mathrm{loc}$ above, 
we can apply the same trick as in \cite[Theorem 4.1]{BeskosEtAl2008}. 
Namely, with $\rho_{\Gamma}(u,v) = \frac{\d P_0(u)}{\d P_\Gamma (u)}(v)$ as given in Theorem \ref{theo:gpCN_density} we obtain
\begin{align*}
	P_\mathrm{loc}(u,\d v)\mu_0(\d u)
	& = \frac 1{\rho_{\Gamma(u)}(u, v)} \, P_0(u, \d v) \mu_0(\d u)\\
	& = \frac {1}{\rho_{\Gamma(u)}(u, v)} \, P_0(v, \d u) \mu_0(\d v)\\
	& = \frac {\rho_{\Gamma(v)}(v, u)}{\rho_{\Gamma(u)}(u, v)} \, P_\mathrm{loc}(v, \d u) \mu_0(\d v),
\end{align*}
where we used the $\mu_0$-reversibility of the pCN proposal $P_0$.
Hence, according to the general Metropolis kernel 
construction outlined in Section \ref{sec: pCN}, 
we have that a Metropolis kernel $M_\mathrm{loc}$ 
with proposal $P_\mathrm{loc}$ and acceptance probability
\begin{equation}\label{equ:aloc}
	\alpha_\mathrm{loc}(u,v) = \min \left\{1,  \exp(\Phi(u)-\Phi(v)) \; \frac {\rho_{\Gamma(u)}(u, v)}{\rho_{\Gamma(v)}(v, u)} \right\}
\end{equation}
is $\mu$-reversible.
Note, that the same construction can analogously be applied to proposals
of the form
\begin{equation}\label{equ:Ploc2}
	P'_\mathrm{loc}(u, \cdot ) = N(\sqrt{1-s^2}u, s^2 C_{\Gamma(u)}),
\end{equation}
where the modified acceptance probability is then given by
\begin{equation}\label{equ:aloc2}
	\alpha'_\mathrm{loc}(u,v) = \min \left\{1,  \exp(\Phi(u)-\Phi(v)) \; 
	\frac {\pi_{\Gamma(u)}(\frac 1s [v-A_0u])}{\pi_{\Gamma(v)}(\frac 1s [u-A_0v])} 
	\right\}
\end{equation}
with $\pi_\Gamma$ as stated in Theorem \ref{theo:gpCN_density}.
The arguments above show that this type of algorithms are well-posed
in infinite dimensions.
Of course, the question arises if the  
additional computational costs of evaluating $\Gamma(u)$ 
and $\rho_{\Gamma(u)}$ or $\pi_{\Gamma(u)}$ in each 
step pay off in a significantly higher statistical efficiency.
Related to this concern, one could think of 
substituting $\nabla G(u)$ in \eqref{eq: local_G} by
a cheaper approximation in order to reduce the computational work. 
This might help to make MH algorithms 
with local proposal covariances feasible.
Unfortunately, the tools and results developed and presented in 
Section~\ref{sec:Conv} are not sufficient to prove spectral gaps 
of these MH algorithms with state-dependent proposals. 
The main reason for this is the missing reversibility of the proposals w.r.t. $\mu_0$. 
This condition played a key role in Theorem \ref{theo:comparison_gap} 
and is the main reason why the analysis of Section~\ref{sec:Conv} is not applicable.
We leave this open for future research.

\subsection*{Acknowledgement}
We thank Oliver Ernst and Hans-J\"org Starkloff for fruitful discussions and valuable
comments. D.R. was supported by the DFG priority program 1324
and the DFG Research training group 1523. B.S. was supported by the DFG priority program 1324.

\appendix
 
\section*{Appendix}

\section{Gaussian measures} \label{sec:Gaussian}

The following brief introduction to Gaussian measures 
is based on the presentations given in \cite[Section~1]{DaPrato2004}
and \cite[Section~3]{Hairer2009}.
Another comprehensive reference for this topic is \cite{Bogachev1998}.

Let $\mc H$ be a Hilbert space with norm $\|\cdot\|$ and inner-product 
$\langle\cdot, \cdot \rangle$ and let $\mc L^1_+(\mc H)$ 
denote the set of all linear, bounded, self-adjoint, positive 
and trace class operators $A:\mc H \to \mc H$.

Let $\mu$ be a measure on $(\mc H, \mc B(\mc H))$ and for simplicity let us assume that
$\int_{\mc H} \|v\|^2\, \mu(\d v) < \infty$.
The \emph{mean} $m \in \mc H$ of $\mu$ is defined 
as the Bochner integral $m = \int_{\mc H} v\, \mu(\d v)$ 
and the \emph{covariance} of $\mu$ is the unique operator 
$C\in\mc L^1_+(\mc H)$ given by
\[
	\langle C u, u' \rangle = \int_{\mc H} \langle u , v - m \rangle \langle u' , v - m \rangle \mu(\d v), \qquad \forall u, u' \in \mc H.
\]
A measure $\mu$ on $\mc H$ is called a \emph{Gaussian measure} 
with mean $m\in \mc H$ and covariance operator $C\in\mc L^1_+(\mc H)$,
denoted by $N(m,C)$, iff 
\[
	\int_{\mc H} \e^{\i \langle u, v \rangle} \mu(\d v) = \e^{\i \langle m , u \rangle - \frac 12\langle C u,u\rangle}, \qquad \forall u \in \mc H.
\]
This definition is equivalent to $\langle u \rangle_*\mu = N(\langle u, m\rangle, \langle Cu, u\rangle)$ for all $u\in\mc H$ where $\langle u \rangle: \mc H \to \bbR$ with $\langle u \rangle(v) := \langle u, v\rangle$ and where $\langle u \rangle_*\mu$ denotes the pushforward measure of $\mu$ under the mapping $\langle u \rangle$.
Gaussian measures are uniquely determined by their mean and covariance, i.e., for any $m\in\mc H$ and any $C \in \mc L^1_+(\mc H)$ there exists a unique Gaussian measure $\mu=N(m,C)$ on $\mc H$.
Moreover, the set of random variables on $\mc H$ distributed according to a Gaussian measure is closed w.r.t. affine transformations.
In detail, let $X \sim N(m,C)$ be a Gaussian randon variable on $\mc H$ and let $b\in\mc H$ and $T\colon \mc H \to \mc H$ be a bounded, linear operator, then due to 
\cite[Proposition 1.2.3]{DaPrato2004} we have
\begin{equation} \label{equ:Gaussian_affine}
	b+TX \sim N(b+Tm, TCT^*).
\end{equation}

The \emph{Cameron-Martin space} $\mc H_\mu$ of a Gaussian measure $\mu = N(m,C)$ on $\mc H$ is defined as  the image space $\Image C^{1/2}$ which forms equipped with $\langle u, v \rangle_{C^{-1}}:= \langle C^{-1/2}u, C^{-1/2}v\rangle$ again a Hilbert space.
The space $\mc H_\mu$ has some surprising properties: 
it is the intersection of all measurable linear subspaces $\mc X \subseteq \mc H$ with $\mu(\mc X)=1$;
if $\ker C = \{0\}$ then $\mc H_\mu$ is dense in $\mc H$ and if $\mc H$ is infinite dimensional then $\mu(\mc H_\mu)= 0$.
Moreover, the space $\mc H_\mu$ plays an important role for the equivalence of Gaussian measures as rigorously expressed in the Cameron-Martin theorem below.
Before stating the result we need some more notation.

In the following let $\mu = N(0,C)$.
For $u\in\mc H_\mu$ we set 
\[
	 W_u(v) := \langle C^{-1/2} u, v \rangle, \qquad \forall v\in \mc H,
\]
and understand $W_u$ as an element of $L_2(\mu)$.
Since the mapping $\mc H_\mu \ni u \mapsto W_u \in L_2(\mu)$ is an isometry \cite[Section 1.2.4]{DaPrato2004}, we can define for any $u\in \mc H$
\[
	 \langle C^{-1/2} u, \cdot \rangle := L_2(\mu)\text{-}\lim_{n\to \infty} W_{u_n}
\]
where $u_n \in \mc H_\mu$ and $u_n \to u$ in $\mc H$ as $n\to \infty$.
And by \cite[Proposition 1.2.7]{DaPrato2004} it holds that
\begin{equation} \label{equ:int_Wu}
	\int_{\mc H} \e^{\langle C^{-1/2} u, v \rangle}\, \mu(\d v) = \e^{\frac 12 \|u\|^2}, \qquad \forall u\in\mc H.
\end{equation}
Hence, if $h\in\mc H_\mu$, we understand $\langle C^{-1} h, \cdot \rangle$ as $\langle C^{-1/2} (C^{-1/2} h), \cdot \rangle \in L_2(\mu)$.

\begin{theo}[Cameron-Martin formula, {\cite[Theorem 1.3.6]{DaPrato2004}}] 
\label{thm: Cameron_Martin_form}
Let $\mu = N(0, C)$ and $\mu_h = N(h, C)$ be Gaussian measures on a separable Hilbert space $\mc H$.
Then, $\mu$ and $\mu_h$ are equivalent iff $h\in \mc H_\mu = \Image C^{1/2}$ in which case
\[
	\frac{\d \mu_h}{\d \mu}(v) = \exp\left( -\frac 12 \|C^{-1/2}h\|^2 + \langle C^{-1} h, v  \rangle \right).
\]
\end{theo} 
Thus, two Gaussian measures $N(m, C)$ and $N(m+h, C)$ are only equivalent 
if $h\in\Image C^{1/2}$.
Consider now $\mu = N(0,C)$ and $\nu = N(0, Q)$ with $C\neq Q$.
Before stating a theorem about the 
equivalence of $\mu$ and $\nu$, we need some more notations.
Let $T:\mc H \to \mc H$ be in the following a self-adjoint trace class operator and let $(t_n)_{n\in\bbN}$ denote the sequence of its eigenvalues. 
We set
\[
	\det(I+T) := \prod_{n=1}^\infty (1+t_n)
\]
and define
\[
	\langle TC^{-1/2}u, C^{-1/2}u\rangle := \lim_{N\to\infty} \langle TC^{-1/2} \, \Pi_N u, C^{-1/2}\, \Pi_N u\rangle, \qquad \mu\text{-a.e.}
\]
where $\Pi_N$ denotes the projection operator to $\mathrm{span}\{e_1,\ldots,e_N\}$ 
with $e_n$ denoting the $n$th eigenvector of $C$.
The existence of the $\mu\text{-a.e.}$-limit above 
is proven in \cite[Proposition 1.2.10]{DaPrato2004} 
and, furthermore, if $\langle Tu, u\rangle < \|u\|^2$ holds for any $u\in\mc H$, 
then by \cite[Proposition 1.2.11]{DaPrato2004} we have
\begin{equation} \label{equ:int_Tu}
	\int_{\mc H} \e^{ \frac 12 \langle T C^{-1/2} u, C^{-1/2} u \rangle)}\, \d \mu(u) = \frac 1{\sqrt{\det(1-T)}}.
\end{equation}

\begin{theo}[{\cite[Proposition 1.3.11]{DaPrato2004}}] 
\label{thm: Feldman_form}
Let $\mu = N(0, C)$ and $\nu = N(0, Q)$ be Gaussian measures on a separable Hilbert space $\mc H$.
If $T:= I - C^{-1/2}QC^{-1/2}$ is self-adjoint, 
trace class and satisfies $\langle Tu, u\rangle < \|u\|^2$ 
for any $u\in\mc H$, then $\mu$ and $\nu$ are equivalent with
\[
	\frac{\d \nu}{\d \mu}(u) = \frac 1{\sqrt{\det(I-T)}} \exp\left( - \frac 12 \langle T(I-T)^{-1}\, C^{-1/2} u, C^{-1/2} u  \rangle \right), \quad u\in\mc H.
\]
\end{theo}

We note that the assumptions of Theorem \ref{thm: Feldman_form} can be relaxed to $I - C^{-1/2}QC^{-1/2}$ being Hilbert-Schmidt which is known as Feldman-Hajek theorem.
Also in this case expression for the Radon-Nikodym derivative can be obtained, see \cite[Corollary 6.4.11]{Bogachev1998}.

Finally, we recall two simple but useful facts resulting from a change of variables.
\begin{lem} \label{lem: change_of_variables}
Let $\mc H$ be a separable Hilbert space, $0<s<\infty$ and $h\in \mc H$.
\begin{itemize}
\item
Assume $\mu = N(m,C)$, $\nu = N(m+h, s^2C)$ on $\mc H$ and $f:\mc H \to \bbR$. Then
\[
	\int_{\mc H} f(v) \mu(\d v) = \int_{\mc H} f\left( \frac 1s (v - h)\right) \, \nu(\d v).
\]

\item
Assume $\mu_1 = N(m_1, C_1)$ and $\mu_2 = N(m_2, C_2)$ are equivalent with $\frac{\d \mu_2}{\d \mu_1}(u) = \pi(u)$.
Then the measures $\nu_1 = N(m_1+h, s^2C_1)$ and $\nu_2 = N(m_2+h, s^2C_2)$  are also equivalent with
\[
	\frac{\d \nu_2}{\d \nu_1}(u) = \pi\left( \frac{u-h}s\right).
\]
\end{itemize}
\end{lem}

\section{Proofs} \label{sec:proofs}

The following proofs are rather operator theoretic 
and rely heavily on the holomorphic functional calculus. 
We refer to \cite[Section VII.3]{DunfordSchwartz1958} 
for a comprehensive introduction.

\subsection{Proof of Lemma \ref{lem:A}} \label{sec:proof_lem:A}
From the proof of Proposition \ref{propo:tildeC} we know 
that $(I + H_\Gamma)^{-1}:\mc H \to \mc H$ is self-adjoint and that $\|(I + H_\Gamma)^{-1}\| \leq 1$.
Thus, $I-s^2(I+H_\Gamma)^{-1}$ is also a self-adjoint, bounded 
and positive operator on $\mc H$ and its square root operator appearing 
in \eqref{equ:A} exists. 
This yields the well-definedness of $A_\Gamma: \mathrm{Im}\,C^{1/2}\to\mc H$. 
We now prove that $A_\Gamma$ is a bounded operator on $\mathrm{Im}\,C^{1/2}$.
For $s=0$ we get $A_\Gamma = I$ and the assertion follows, so that we assume $s\in(0,1)$.
Let us now define $f\colon \mathbb{C}\setminus \{-1\} \to \mathbb{C}$ by 
\[
	f(z) = \sqrt{1 - s^2(1+z)^{-1}}.
\]
The function $f$ is analytic in the complex half plane $\{z \in \bbC: \Re(z) > s^2-1\}$, since $\Re(1+z) > s^2$ implies
\[
	\Re\left( (1+z)^{-1} \right) = \frac{\Re(1+z)}{|1+z|^2} \leq \frac1{\Re(1+z)}< \frac 1{s^2}.
\]
Denoting $\gamma:=\|H_\Gamma\|$ the spectrum of $H_\Gamma = C^{1/2}\Gamma C^{1/2}$ is contained in $[0,\gamma]$.
Then, since $s<1$ we have that $f$ is analytic in a neighborhood, say, $\mathcal{N}[0,\gamma]$ of $[0,\gamma]$. 
Hence, by functional calculus we obtain
\[
	\sqrt{I - s^2 \left( I + H_\Gamma \right)^{-1} } 
	= f(H_\Gamma) 
	= \frac 1{2\pi i} \int_{\partial \mc N[0,\gamma]} f(\zeta)\, (\zeta I - H_\Gamma)^{-1} \, \d \zeta.
\]
Due to analyticity we can approximate $f$ by a sequence of polynomials $p_n$ with degree $n$ which converge uniformly on $\mathcal{N}[0,\gamma]$ to $f$ for $n\to\infty$.
Then, by \cite[Lemma VII.3.13]{DunfordSchwartz1958} holds
\[
	\|p_n(H_\Gamma) - f(H_\Gamma) \|_{\mc H\to \mc H} \to 0,
\]
for $n \to \infty$.
Since the polynomials $p_n$ can be represented as $ p_n(z) = \sum_{k=0}^n a_k^{(n)} z^k$, we obtain further
\[
	C^{1/2} \, p_n(H_\Gamma) 
	= C^{1/2}  \sum_{k=0}^n a^{(n)}_k  \, (C^{1/2}\Gamma C^{1/2})^k = p_n(C\Gamma) \, C^{1/2}.
\]
By \cite[Proposition~1]{Hladnik88} we have
\[
 \spec(C\Gamma\mid \mc H) 
 = \spec(C^{1/2}\Gamma C^{1/2} \mid \mc H)\subseteq [0,\gamma]
\]
where $\spec(\cdot\mid \mc H)$ denotes the spectrum on $\mc H$, and, thus, we can conclude
$\|p_n(C\Gamma) - f(C\Gamma) \|_{\mc H \to \mc H} \to 0$ as $n \to \infty$ again by \cite[Lemma VII.3.13]{DunfordSchwartz1958}.
Hence,
\begin{align*}
	C^{1/2} f(H_\Gamma) 
	& = \lim_{n\to\infty} C^{1/2} \, p_n(H_\Gamma)
	= \lim_{n\to\infty} p_n(C\Gamma) \, C^{1/2}
	= f(C\Gamma) C^{1/2}
\end{align*}
and
\[
	A_{\Gamma}= C^{1/2} f(H_\Gamma) C^{-1/2} = f(C\Gamma) C^{1/2}C^{-1/2}  = f(C\Gamma)
\]
where $f(C\Gamma)$ is by construction a bounded operator on $\mc H$. \hfill $\qed$

\subsection{Proof of Lemma \ref{lem:Delta_Gamma}} \label{sec:proof_lem:Delta_Gamma}
By \cite[Theorem 1]{Douglas1966} the relation $\Image(\Delta_\Gamma) \subseteq \Image(C^{1/2})$ holds iff there exists a bounded operator $B:\mc H\to\mc H$ such that
\begin{equation} \label{equ:ACB}
	\Delta_\Gamma = C^{1/2} B.
\end{equation}
Thus, $\Image(\Delta_\Gamma) \subseteq \Image(C^{1/2})$ 
is equivalent to $C^{-1/2}\Delta_\Gamma$ being bounded on $\mc H$.
In order to construct and analyze the operator $B$, 
we define $f\colon \bbC\setminus\{-1\} \to \bbC$ by
\[
	f(z) := \sqrt{1 - s^2(1+z)^{-1}} - \sqrt{1-s^2}, 
\]
which is analytic in $\{z \in \bbC: \Re(z) > s^2-1\}$, 
cf. the proof of Lemma \ref{lem:A}, and particularly in
\[
	V = \{z \in \bbC: \dist(z, [0,\gamma]) \leq \varepsilon\}, \qquad 0 < \varepsilon < 1-s^2,
\]
where $\gamma := \|H_\Gamma\|$.
We have the following representation
\begin{align*}
	 -\Delta_\Gamma & = A_{\Gamma} - \sqrt{1-s^2} I\\
	& = C^{1/2}\,\left( \sqrt{I - s^2 \left( I + H_\Gamma \right)^{-1} } - \sqrt{1-s^2} I\right) \, C^{-1/2}\\
	& = C^{1/2}\, f(H_\Gamma) \, C^{-1/2}
\end{align*}
with
\[
	f(H_\Gamma) = \frac 1{2\pi i} \int_{\partial V} f(\zeta)\, (\zeta I - H_\Gamma)^{-1} \, \d \zeta
\]
see \cite[Chapter VII.3]{DunfordSchwartz1958}.
Hence, if we can prove that $B = - f(H_\Gamma) \, C^{-1/2}$ 
is a bounded operator on $\mc H$, we have shown the assertion.

For this let $p_n(z) = \sum_{k=0}^n a_k^{(n)} z^k$ be polynomials of degree $n$, with $n\in\bbN$, which converge uniformly on $V$ to $f$.
Such polynomials exist due to the analyticity of $f$ and 
by the fact that $f(0) = 0$ we can assume w.l.o.g. that $a_0^{(n)} = 0$ for all $n\in\bbN$.
This leads to
\begin{align*}
	p_n(H_\Gamma) 
	& =  C^{1/2} \Gamma^{1/2} \left( \sum_{k=1}^n a_k^{(n)} (\Gamma^{1/2} C \Gamma^{1/2} )^{k-1} \right) \Gamma^{1/2} C^{1/2}\\
	& = C^{1/2} \Gamma^{1/2} \; q_{n-1}(\Gamma^{1/2} C \Gamma^{1/2}) \; \Gamma^{1/2} C^{1/2}
\end{align*}
with $q_{n-1}(z) := \sum_{k=1}^{n} a_{k}^{(n)} z^{k-1} = p_n(z)/z$.
Now, \cite[Proposition~1]{Hladnik88} implies that the operators $C^{1/2}\Gamma C^{1/2}$ and $\Gamma^{1/2}C \Gamma^{1/2}$ share the same spectrum, since $C$ and $\Gamma$ are positive.
Thus, $\spec(\Gamma^{1/2}C \Gamma^{1/2}\mid \mc H)  \subset [0,\gamma]$ and we have
\[
	q_n(\Gamma^{1/2}C \Gamma^{1/2})	= \frac 1{2\pi i} \int_{\partial V} q_n(\zeta)\, (\zeta I - \Gamma^{1/2}C \Gamma^{1/2})^{-1} \, \d \zeta, \qquad n\in\bbN.
\]
Moreover, the polynomials $q_n$ are a Cauchy sequence in $C(\partial V)$, since
\begin{align*}
	\sup_{\zeta \in \partial V}|q_n(\zeta) - q_m(\zeta)| 
	& \leq \sup_{\zeta \in \partial V}\frac {|\zeta|}{\min_{\eta \in \partial V} |\eta|} |q_n(\zeta) - q_m(\zeta)| \\
	& = \frac {1}{\min_{\eta \in \partial V} |\eta|} \sup_{\zeta \in \partial V}|\zeta q_n(\zeta) - \zeta q_m(\zeta)| \\
	& = \frac {1}{\min_{\eta \in \partial V} |\eta|} \sup_{\zeta \in \partial V}|p_{n+1}(\zeta) - p_{m+1}(\zeta))|
\end{align*}
where $\min_{\eta \in \partial V} |\eta| = \varepsilon > 0$ due to our choice of $V$.
Thus, the polynomials $q_n$ converge uniformly on $\partial V$ to a function $g$. 
This implies that the operators 
$q_n(\Gamma^{1/2}C \Gamma^{1/2})$
converge in the operator norm to a bounded operator
\[
	g(\Gamma^{1/2}C \Gamma^{1/2}) := \frac 1{2\pi i} \int_{\partial V} g(\zeta)\, (\zeta I - \Gamma^{1/2}C \Gamma^{1/2})^{-1} \, \d \zeta.
\]
We arrive at
\begin{align*}
	f(H_\Gamma) & = \lim_{n\to\infty} p_n(C^{1/2}\Gamma C^{1/2}) \\
	& = \lim_{n\to\infty} C^{1/2} \Gamma^{1/2} \; q_{n-1}(\Gamma^{1/2} C \Gamma^{1/2}) \; \Gamma^{1/2} C^{1/2} \\
	& =  C^{1/2} \Gamma^{1/2} \; g(\Gamma^{1/2} C \Gamma^{1/2}) \; \Gamma^{1/2} C^{1/2}, 
\end{align*}
which yields
\[
	B = -f(H_\Gamma) C^{-1/2} = -C^{1/2} \Gamma^{1/2} \; g(\Gamma^{1/2} C \Gamma^{1/2}) \Gamma^{1/2}
\]
being bounded on $\mc H$. \hfill $\qed$


\end{document}